\documentclass{article}
\usepackage{amsmath,amsbsy,amsfonts,amssymb,amsthm,amscd,amsxtra,amstext,amsopn}
\usepackage{mathtools,emptypage,dsfont,yfonts,latexsym,mathrsfs,bbold}
\usepackage[left=0.6in,right=1.1in]{geometry}
\usepackage[fontsize=11pt]{scrextend}
\usepackage{bclogo}%  \bcpanchant
\usepackage{xcolor}
\usepackage{comment}
\usepackage{marginnote}

\def\red{}

\newcommand{\scrF}{{\mathscr F}}
\newcommand{\scrG}{{\mathscr G}}

\newcommand\calD{{\mathcal D}}
\newcommand \ul[1]{\underline{#1}}
\renewcommand\v[1]{\vec{#1}}

\newcommand \Id{\mathrm{Id}}

\newcommand*\norm[1]{ \left\| #1 \right\| }

\newcommand*\R{ \mathds{R} }

\newcommand*\EXP{ \mathds{E} }

\newcommand*\cL{ \mathcal{L} }

\newcommand\set[1]{\left\{ #1 \right\}}

\newcommand\uv{{\underline v}}
\newcommand\uu{{\underline u}}
\newcommand\uw{{\underline w}}
\newcommand\uxi{{\underline \xi}}
\newcommand\ueta{{\underline \eta}}
\newcommand\uzeta{{\underline \zeta}}
\newcommand\vv{{\vec v}}
\newcommand\vu{{\vec u}}
\newcommand\vw{{\vec w}}
\newcommand\vxi{{\vec \xi}}
\newcommand\veta{{\vec \eta}}

\renewcommand\S{{\mathds S}}

\newtheorem{mythm}{Theorem}[section]

\newtheorem{mylemma}[mythm]{Lemma}

\newtheorem{corollary}[mythm]{Corollary}

\newtheorem{myclaim}{Claim}

\newtheorem{myrmk}{Remark}

\title{\bf \huge A Kac system interacting with two heat reservoirs}

\author{{\bf Federico Bonetto$^{1,*}$, Michael Loss$^1$ and Matthew Powell$^1$}\\
\small $^1$ School of Mathematics, Georgia Institute of Technology,Atlanta, GA 30332, USA\\
\small $^*$ corresponding author: \tt bonetto@math.gatech.edu}

\date{\today}

\begin{document}
\maketitle

\begin{abstract}
 We study a system formed by $M$ particles moving in 3 dimension and interacting with 2 heat reservoirs with $N>>M$
particles each. The system and the reservoirs evolve and interact via random collision described by a Kac-type master
equation. The initial state of the reservoirs is given by 2 Maxwellian distributions at temperature $T_+$ and $T_-$. We
show that, for times much shorter than $\sqrt{N}/M$ the interaction with the reservoirs is well approximated by the
interaction with 2 Maxwellian thermostats, that is, heat reservoirs with $N=\infty$. As a byproduct, if $T_+=T_-$ we
extend the results in \cite{BLTV} to particles in 3 dimension.
\end{abstract}

\begin{center} {\bf Acknoledgement}\\
	FB acknowledge partial support from U.S. National Science Foundation grant DMS-1907643. ML acknowledge partial 
	support from U.S. National Science Foundation grant DMS-2154340.
\end{center}

\section{Introduction}

In \cite{Kac} Mark Kac introduced a toy model, the ``Kac Evolution",  to study kinetic properties of gases. He
considered a system of $M$ particles moving in one dimension and interacting via random collisions. After an
exponentially distributed time, two particles are randomly and uniformly selected and they suffer a collision; that is,
their energy is randomly redistributed. The rate of collision $\lambda_M$ is chosen in such a way that the average
number of collisions a particle suffers in a unit time is independent of $M$.

Despite its simplicity, the model serves to explain a number of issues. We give a short and informal description before
turning to the main point of the paper. The original purpose of Kac was to derive the spatially homogeneous Boltzmann
equation. This was achieved through the notion of `chaos', i.e., states for which position and velocity of the particles
are asymptotically independent as the particle number tends to infinity and are, thus, characterized by their single
particle marginal. He proved that this property was preserved under the Kac Evolution, a phenomenon known as
`propagation of chaos', and the evolution of the marginal is given by the non-linear Boltzmann-Kac equation. The three
dimensional version of this model is discussed in \cite{KacBook} where the analogous results are proved. The main
limitation of these papers is that the molecules are Maxwellian:  the intensity of the collision process is independent
of the velocities.  This restriction was removed by Sznitman \cite{Sznitman84} (see also the seminal paper
\cite{MischlerMouhot1}).

The Kac Evolution is ergodic and hence it makes sense to quantify the rate of approach to equilibrium. One way to
achieve this is through the gap of the generator, which Kac \cite{Kac} conjectured to be bounded away from zero
uniformly in the particle number. This was eventually proved in \cite{Jeanvresse}. In \cite{CCL1} and \cite{Maslen} an
exact formula for the gap is given. Results for the gap with models with velocity dependent intensity are given in
\cite{villani}, \cite{CCL2}, and \cite{CCL3} (in the latter, Kac's conjecture for three dimensional hard spheres is
proved).

The notion of a gap means that one views the generator of the evolution on a Hilbert space, but it is well known that
probability distributions   have $L^2$ norms that grow exponentially with the particle number. This is remedied by
considering the entropy, which, since it is extensive, is a more natural notion of approach to equilibrium. One
expectation is that the entropy decreases exponentially in time with a rate that is independent of the particle number.
The results, so far, do not lend optimism to this expectation. The best results are those of Villani \cite{villani} and
Einav \cite{amit} where it is shown that the entropy production in the worst case is inversely proportional to the
number of particles.

Recall that the entropy production is the ratio of the entropy dissipation and the entropy and hence, strictly speaking,
the aforementioned results do not make a statement about the time decay of the entropy of a generic state. If, however,
one considers initial states where only a small part of the system is out of equilibrium, then one can show exponential
decay for the entropy as shown in \cite{BGLR} (see also \cite{BCHL}).

A rather different interpretation of the setup in \cite{BGLR}, which is the central issue for the current work,  is that
the part of the system initially in equilibrium can be seen as a reservoir and the remaining part as a system that is
out of equilibrium. In the simplest case one considers a system of $M$ particles interacting via a Kac Evolution with a
reservoir of $N$ particles. Given that the reservoir is initially in equilibrium one expects that for $N$ large the
evolution of this state can be approximated by an effective evolution of a system of $M$ particles interacting with a
thermostat, i.e., one replaces the reservoir by a process in which one generate a ``virtual" particle randomly
selected with the same distribution of the initial state of the reservoir and lets it collide with a particle in the
system. Thus, the idea is that the distribution of the reservoir, for large $N$, stays essentially fixed along the
evolution. This approximation was carried out in \cite{BLTV} for particles moving in one dimension uniformly in time,
using the Gabetta-Toscani-Wennberg metric.

\red{The present paper mainly addresses two issues. Of interest is an extension of the work in [3] to a system of M 
particles
that is coupled to two or more reservoirs of N particles each that are initially in equilibrium at {\it different}
temperatures. In Theorem 2.1 we show that up to a time of order $\sqrt N/M$ the evolution of the reservoir system can be
approximated uniformly by the corresponding evolution of the thermostats. It is know (see \cite{CLM,Evans} or Corollary
\ref{Cor:GapDecay} below) that as the time $t \to \infty$ the thermostat system approaches a steady state (which is not
explicitly known) at a rate of order one. As a consequence for times $1<< t << \sqrt N/M$ the reservoir system stays,
approximately, in a non-equilibrium steady state (NESS). One should emphasize that for times $t>> \sqrt N/M$ the
reservoir system tends to an equilibrium (which in this model is given by the average of the initial condition over all
rotations). It should be emphasized that this equilibrium state is not close in any sense to the NESS of the thermostat.
 The physical distinction between these two states, the equilibrium and the non-equilibrium steady state, is that the
latter mediates a heat transfer from the thermostat at higher temperature to the one at lower temperature. This
distinction is absent in the one thermostat case where there is no NESS and the thermostat system tends to an
equilibrium given by a Maxwellian.}

%The present paper mainly addresses two issues. Of interest is an extension of the work in \cite{BLTV} to a system
%of $M$ particles that is coupled to two or more reservoirs of $N$ particles each that are initially in equilibrium
%at different temperatures. One cannot expect to approximate this evolution by one where the reservoirs are replaced
%by thermostats uniformly in time. 
%If the temperatures of the reservoirs are different, they will exchange heat
%through the system. Thus the system together with the reservoirs evolves toward an equilibrium state, which is given by the spherical average of the initial state. 
%\mnote{Discuss timescales and comparison of the steady-states}On the one
%hand, since the temperatures of the thermostats are fixed, the thermostated system cannot have an equilibrium
%state. Indeed it was proven in \cite{CLM,Evans} (see also Corollary \ref{Cor:GapDecay}) that the thermostated system
%instead has a non-equilibrium steady state which is not explicitly given. This is in contrast to the one
%thermostat case where there is an equilibrium given by a Maxwellian. Thus, the time scale over which this
%approximation is valid is finite but grows with $N$, and it is of interest to understand this rate of growth.

The second issue is how to extend the analysis in \cite{BLTV} to three dimensions. There is an obvious difficulty: any
reasonable collision mechanism conserves both energy and momentum; the additional conservation law complicates the analysis.
A somewhat deeper issue is the possibility of considering reservoirs that initially have a non zero average momentum.
This would produce a shear in the system, but in this situation it is not even clear how to implement the relevant
thermostat.

The paper is organized as follows.  In Section \ref{modelandresults} we present the models and state the results.
In Section \ref{newinequ} we state and comment on the key functional inequality.
In Section \ref{sec:out} we comment on our results and outline possible
extension and ideas for future work. Section \ref{sec:proof_prop} contains the proof of Lemmas \ref{lem:prop5},
\ref{lem:stimal} and \ref{lem:prop5new} while in Section \ref{sec:2ther} we report the proof of Theorem \ref{thm:2ther}.
Finally the Appendices contain technical results useful in the proofs.

\section{Model and Results} \label{modelandresults}

The basic ingredient to build the Kac model evolution is the operator that describes the effect of a collision between
two particles. Assume that the particles are spherical and that the unit vector associated with the relative position of
their centers is $\omega$. If $\vec v_1$ and $\vec v_2$, with $\vec v_i=(\vec v_{i,1}, \vec v_{i,2},\vec v_{i,3})\in\R^3$,
are the outgoing velocities just after the collision then the incoming velocities before the collision were
\[
\vec v_1^*(\omega)=\vec v_1-((\vec v_1-\vec v_2)\cdot\omega)\omega\,,\qquad \vec v_2^*(\omega)=\vec
v_2-((\vec v_2-\vec v_1)\cdot\omega)\omega\, ,
\]
where $(\vv\cdot\vw)$ denotes the scalar product in $\mathds R^3$. \red{Assume now that before the collision the
velocities of the two particles are distributed according to the probability distribution $g$ on $\R^6$. Given
$\omega$, the probability of observing the two particles after the collision with velocities $\vec v_1$, and $\vec
v_2$ is $g(\vec v_1^*(\omega),\vec v_2^*(\omega))$. If we further assume that the unit vector $\omega$ is chosen 
uniformly  at random on the unit sphere $\mathds S^2$, that is according to the normalized Haar measure $d\omega$,
the effect of the collision on $g$ is given by}
\begin{equation}\label{eq:defR}
R[g](\vec v_1,\vec v_2)=\int_{\mathds S^2}g(\vec v_1^*(\omega),\vec v_2^*(\omega))d\omega\, .
\end{equation}

With this collision in hand, we can now start to describe the evolution in the models we will consider. The superscript $S$ will refer to the
{\it system}, as opposed to the reservoirs or the thermostats, for the remainder of this paper. The system contains $M$
particles characterized by their velocities $\underline v=(\vec v_1,\ldots,\vec v_M)\in\mathds R^{3M}$. Since the
collision times are randomly chosen, the position of the particles plays no role. Thus, the state of the system is
described by a probability distribution $f(\underline v)$ on $\mathds R^{3M}$. Since the particles are assumed to be
identical, we will only consider {\it permutation invariant} distributions, that is we will assume that, for every 
permutation $\pi$
on $\set{1,\ldots,N}$ we have
\begin{equation}\label{eq:permu}
f(\vv_1,\ldots,\vv_M)=f(\vv_{\pi(1)},\ldots,\vv_{\pi(M)})\, .
\end{equation}

The effect of a collision between particle $i$ and $j$ on the distribution $f$ is described by the operator $R$ in
\eqref{eq:defR} acting on $\vec v_i$ and $\vec v_j$, that is
\begin{equation}\label{eq:RS}
 R_{i,j}^S[f](\underline v)=\int_{\mathds S^2}f(\underline v_{i,j}(\omega))d\sigma(\omega)
\end{equation}
where 
\begin{equation}\label{eq:romega}
\begin{aligned}
 &\underline v_{i,j}(\omega)=(\vec v_1,\ldots,\vec v_i^*(\omega),\ldots,\vec
v_j^*(\omega),\ldots,\vec v_M)\,,\\
 &\vec v_i^*(\omega)=\vec v_i-((\vec v_i-\vec v_j)\cdot\omega)\omega\,,\qquad \vec v_j^*(\omega)=\vec
v_j-((\vec v_j-\vec v_i)\cdot\omega)\omega\,.
\end{aligned}
\end{equation}

We assume that the collisions take place according to a Poisson process of intensity $\lambda_M$ and that when a
collision takes place, the pair of colliding particles is chosen randomly and uniformly. Choosing
$\lambda_M=M\lambda_S/2$ we get that the infinitesimal generator for the evolution of $f$ is
\begin{equation}\label{eq:genLS}
 \mathcal L_S[f]=\frac{\lambda_S}{M-1}\sum_{i<j} (R^S_{i,j}[f]-f).
\end{equation}
In this way the average number of collisions that a particle suffers in a unit time is equal to $\lambda_S$ independently 
of $M$. $\lambda_S^{-1}$ is called the {\it mean free flight} for the system. For simplicity, we will choose a unit of 
time so that $\lambda_S=1$.

In \cite{BLV}, in an attempt to introduce more structure in the Kac model, the authors consider the interaction of a Kac
system with a {\it Maxwellian thermostat}. This is thought of as an infinite reservoir of particles in equilibrium at
temperature $T$.  When a particle interacts with the thermostat it undergoes a collision with a ``virtual'' particle
randomly selected from a Maxwellian distribution at temperature $T$. After the collision, the virtual particle
disappears. If $g(\vec v)$ is the probability distribution of the velocity $\vec v$ of the particle, then the effect of
the interaction with the thermostat is
\begin{equation}\label{eq:defB}
B[g](\vec v)=\int_{\R^3}R[g\Gamma_T](\vec v,\vec w)d\vec w=\int_{\mathds R^3} d\vec w  \int_{\mathds 
S^2}g(\vec v^*(\omega))\Gamma_T(\vec w^*(\omega)) d \omega.
\end{equation}  
with 
\[
\Gamma_T(\vec w)=\left(\frac{1}{2\pi T}\right)^{\frac 32}e^{-\frac1{2T}\|\vec w\|^2}
\]
and
\begin{equation}\label{eq:rbomega}
\vec v^*(\omega)=\vec v_i-((\vec v_i-\vec w)\cdot\omega)\omega\qquad \vec w^*(\omega)=\vec
w-((\vec w-\vec v_i)\cdot\omega)\omega\, .
\end{equation}
At Poisson distributed times, a particle in the system is selected uniformly at random, to interact with the thermostat.
We can represent the action of this thermostat on the $M$-particle distribution $f$ via the infinitesimal generator
\begin{equation}\label{eq:genLB}
 \mathcal L_B[f]=\mu\sum_{i=1}^M (B_i[f]-f)\, .
\end{equation}
where $B_i$ is the operator \eqref{eq:defB} acting on the velocity $\vec v_i$ of the $i$-th particle.
Thus the evolution of the combined system+thermostat can be described via the generator
\begin{equation}\label{eq:gen1ther}
\widetilde{ \mathcal L}_T=\mathcal L_S+ \mathcal L_B.
\end{equation}

For the one dimensional version of this evolution, it is possible to show that for every given initial distribution
$f_0$, the evolved distribution $f_t:=e^{\widetilde{\mathcal L}_T t}f_0$ converges exponentially fast to a Maxwellian
distribution both in a suitable $L^2$ norm and in entropy, see \cite{BLV}. The $L^2$ norm is not a very useful measure
of approach to equilibrium when $M$ is large, see \cite{BLV} for more details, while typical alternatives like the
entropy are very hard to analyze. For these reasons, in this paper, we will work with the GTW distance $d_2$ defined as
follows (see also \cite{GTW}). Given two functions $f,g\in L^1(\mathds{R}^{3M})$ with
\[
\int_{\mathds{R}^{3M}} f(\underline v)d\underline v=\int_{\mathds{R}^{3M}} g(\underline v)d\underline v \qquad
\int_{\mathds{R}^{3M}} \underline v f(\underline v)d\underline v= 
\int _{\mathds{R}^{3M}}\underline v g(\underline v)d\underline v
\]
we set
\begin{equation}\label{eq:d2}
 d_2(f,g)=\sup_{\underline \xi\not=0}\frac {|\hat f(\underline \xi)-\hat g(\underline \xi)|}{|\underline \xi|^2}
\end{equation}
where $\hat f(\uxi)$ is the Fourier transform of $f(\underline v)$ defined as
\[
 \hat f(\ul\xi)=\int_{\mathds R^{3M}}e^{2\pi i (\underline\xi \cdot \underline v)}f(\underline v)d\underline v
\]
and $(\underline\xi \cdot \underline v)=\sum_{i=1}^M (\v v_i\cdot\v \xi_i)$ is the scalar product in $\mathds R^{3M}$.
It is not hard to show that the evolution generated by $\widetilde{\mathcal L}_T$ converges exponentially to equilibrium 
also in the
$d_2$ distance, see \cite{CLM,Evans} and Appendix \ref{app:basic} below.

It is natural to see the Maxwellian thermostat as an idealization of the interaction of our small system with $M$
particles with a much larger reservoir with $N\gg M$ particles initially in equilibrium at temperature $T$. We now
describe the combined system+reservoir evolution for the three dimensional version of the model we will use in this
paper. The state of the combined system is described by a probability distribution $F(\uv, \uw)$ on $\mathds R^{3M+3N}$.
The generator is given by
\begin{equation}\label{eq:gen1res}
 \mathcal L_T=\mathcal L_S +\mathcal  L_R + \mathcal L_I
\end{equation}
where $\mathcal L_R$ describes the evolution inside the reservoir. This is modeled as a standard Kac evolution analogous
to that of the system with possibly a different mean free flight $\lambda_R^{-1}$. This means that
\begin{equation}\label{eq:genLR}
 \mathcal L_R[F]=\frac{\lambda_R}{N-1}\sum_{i<j} (R_{i,j}^R[F]-F)
\end{equation}
with $R_{i,j}^R$ describing the effect of a collision in the reservoir between particle $i$ with velocity $\v w_i$ and
particle $j$ with velocity $\v w_j$ via the operator \eqref{eq:defR} acting on $\vec w_i$ and $\vec w_j$. Finally
$\mathcal L_I$ describes the interaction between the system and the reservoir. It is also modeled via a Kac style
collision and we set
\begin{equation}\label{eq:genLI}
 \mathcal L_I[F]=\frac \mu N\sum_{i=1}^N\sum_{j=1}^M (R_{i,j}^I[F]-F)
\end{equation}
where $R_{i,j}^I$ describes the effect of a collision between particle $i$ in the reservoir with velocity $\v w_i$ and
particle $j$ in the system with velocity $\v v_j$. It is again given by the operator \eqref{eq:defR} acting on $\vec
v_i$ and $\vec w_j$. Observe that the factor $\mu/N$ in \eqref{eq:genLI} has the effect that a particle in the system
suffers, in average, $\mu$ collisions per unit time with particles in the reservoir. At the same time, a particle in the
reservoir on average suffers $\mu M/N$ collision per unit time with particles in the system.

Since the reservoir is assumed to be initially in equilibrium, we choose an initial state of the form
\begin{equation}\label{eq:ini1}
 F_0(\underline v,\underline w)=f_0(\underline v)\Gamma_T^N(\underline w)
\end{equation}
where
\[
 \Gamma_T^N(\underline w)=\prod_{i=1}^N \Gamma_T(\vec w_i)
\]
is the Maxwellian distribution at temperature $T$ for $N$ particles. 

For the evolution in one dimension, the relation between the finite reservoir and the thermostat was studied in
\cite{BLTV}. The authors prove there that, for $N$ large, the combined evolution of system+reservoir is well approximated
by the system+thermostat idealization uniformly in time in the $d_2$ distance. A generalization of this results to 
the 
evolution in 3 dimension will be a corollary of our main result that we are going to present now.

As discussed above, to obtain a truly out of equilibrium evolution, we consider the situation when a small system 
with $M$
particles interacts with two large reservoirs initially at temperature $T_+$ and $T_-$ with $N_+$ and $N_-$ 
particles
respectively. The state of the combined system is now given by a distribution $\scrF(\uu,\uv,\uw)$ on
$\R^{3(M+N_-+N_+)}$ and, analogously to \eqref{eq:ini1}, the initial state is of the form
\begin{equation}\label{eq:ini2}
 \scrF_0(\uu,\uv,\uw)=\Gamma_+(\uu)f_0(\uv)\Gamma_-(\uw)
\end{equation}
where $\Gamma_\sigma=\Gamma_{T_\sigma}^{N_\sigma}$ with $\sigma=\pm$. \red{We also assume that the average velocity of the center of mass of the system vanishes while the initial temperature is finite, that is 
\[
\int_{R^{3M}} \ul v f_0(\ul v)d\ul v=0, \qquad\hbox{and}\qquad T_S:=\frac1{3M}\int_{R^{3M}}\|\ul v\|^2 f_0(\ul v)< \infty\, .
\]}
We can write the generator of the evolution as
\begin{equation}\label{eq:gen2res}
 \mathcal L=\mathcal L_S+\mathcal L_{R_+}+\mathcal L_{R_-}+\mathcal L_{I_+}+\mathcal L_{I_-}
\end{equation}
where $\mathcal L_S$ still describes the internal evolution of the system, see \eqref{eq:genLS}, while $\mathcal
L_{R_\sigma}$ and $\mathcal L_{I_\sigma}$, $\sigma=\pm$, describe the internal evolution of the $\sigma$ reservoir
and the interaction of the $\sigma$ reservoir with the system, respectively, see \eqref{eq:genLR} and \eqref{eq:genLI}.
Observe that there is no direct interaction between the two reservoirs.

If $T_+=T_-$ the two reservoir are essentially equivalent to a single reservoir. This situation will be discussed in
Corollary \ref{thm:1ther} below. We will thus always assume that $T_+>T_-$. On the other hand, to simplify 
notation, we
will assume that $N_+=N_-=N$. Thus we will consider the evolution
\begin{equation}\label{eq:evoR}
\scrF_t:=e^{\mathcal L t}\scrF_0
\end{equation}
generated by the generator $\mathcal L$ in \eqref{eq:gen2res} starting from the initial condition $\scrF_0$ in 
\eqref{eq:ini2}.

It is natural to compare the evolution of the combined system+2 reservoirs described in \eqref{eq:evoR} with the
evolution of a  small system with $M$ particles interacting with two Maxwellian thermostats at temperature $T_+$ and
$T_-$. In such a situation the generator of the evolution is given by %
\begin{equation}\label{eq:gen2ter}
\widetilde{\mathcal L}=\mathcal L_S+\mathcal L_{B_+}+\mathcal L_{B_-}
\end{equation}
where $\mathcal L_{B_\sigma}$ describes the thermostat at temperature $T_\sigma$ with $\sigma\in\{+,-\}$, see
\eqref{eq:genLB}, and the initial state is given by $f_0$. That is we want to find an upper bound for the $d_2$ distance
between $\scrF_t$ in \eqref{eq:ini2} and
\[
 \widetilde \scrF_t:=\Gamma_+
e^{\widetilde{\mathcal L}t}f_0\Gamma_-=:\Gamma_+\tilde f_t\Gamma_-
\]
in term of the temperature difference $T_+-T_-$, the $d_2$ distance between $f_0$ and $\Gamma^M_\pm$, and the
moments of $f_0$. \red{ Before stating our result we observe that, since $\cL$ and $\widetilde \cL$ are bounded operators on $L^1(\R^{3(M+2N)})$ and $L^1(\R^{3M})$, repsectively, both $\scrF_t$ and $\widetilde\scrF_t$ exists for every $t$. Moreover they are probability distributions with vanishing first moments and finite second moments so that $d_2(\scrF_t,\widetilde\scrF_t)$ is well defined and finite. }

To formulate our main Theorem we need a few more definitions. Let $I:=\set{(i,j):i=1,\ldots,M\,,\,j=1,2,3}$ be the
set of indices for the components of the vector $\uv$ and, for $k>0$, consider the set of monomials $\uv_{\bf
	i}=\prod_{l=1}^k v_{i_l,j_l}$ where $\mathbf i=((i_1,j_1),\ldots,(i_k,j_k))\in I^k$. Given a probability
distribution $f$ on $\mathds R^{3M}$ we define %its $\mathbf i$-th moment as
\begin{equation}\label{eq:mi}
e_{\mathbf i}(f):=\int_{\mathds R^{3M}} |\underline v_{\mathbf i}|\,f(\uv)d\uv
\end{equation}
and we set
\begin{equation}\label{eq:E4}
	E_4(f):=\max_{\mathbf i\in\tilde I^4} e_{\mathbf i}(f)
\end{equation}
where $\tilde I^4=\bigcup_{k=0}^4 I^k$.

We are now ready to formulate our results that are contained in the following Theorem.

\begin{mythm}\label{thm:2ther}
Let $f_0(\underline v)$ be a probability distribution on $\mathds R^{3M}$ with \red{$\int \ul v f_0(\ul v)d\ul v=0$ and} $E_4(f_0)<\infty$ and let
$\scrF_0(\underline v,\underline u,\underline w)=\Gamma_+(\underline u)f_0(\underline v)\Gamma_-(\underline w)$ as
discussed in \eqref{eq:ini2}. Consider the two evolved distributions on $\mathds R^{3(2N+M)}$ given by
\[
 \scrF_t=e^{\mathcal L t}\scrF_0\qquad\qquad
 \widetilde \scrF_t=\Gamma_+ e^{\widetilde{\mathcal L}_T t}f_0\Gamma_-
\]
with $\mathcal L$ and $\widetilde{\mathcal L}$ given by \eqref{eq:gen2res} and \eqref{eq:gen2ter} respectively. Then we
have
\begin{equation}\label{eq:stima2}
 d_2(\scrF_t,\widetilde \scrF_t)\leq 
%C \frac{M}{\sqrt N}E_4(f_0)^{\frac56}(1-e^{-\frac\mu9t})
%(d_2(f_0,\Gamma^M_{+})^{\frac 16}+d_2(f_0,\Gamma^M_{-})^{\frac 16})+
%C E_4(f_0)^{\frac56}\frac{M}{\sqrt N}(T_+-T_-)^{\frac 16}t
C \frac{M}{\sqrt N}E_4(f_0)^{\frac56}\left((1-e^{-\frac\mu9t})\left(d_2(f_0,\Gamma^M_{+})^{\frac 
16}+d_2(f_0,\Gamma^M_{-})\right)^{\frac 16}+
(T_+-T_-)^{\frac 16}t
\right)
\end{equation}
for a suitable constant $C$.
\end{mythm}

\begin{myrmk}\emph{
Here and in what follows, $C$ will indicate a generic constant, independent of
$N$, $M$ and $f_0$ whose numerical value is not relevant and may vary even
inside the same formula.}
\end{myrmk}

Due to the presence of the factor of $t$ in the last term of \eqref{eq:stima2}, the above theorem tells us nothing about
the difference between $\scrF_t$ and $\widetilde \scrF_t$ when $t$ is very large.  This is not strange since, as opposed
to the case of only one reservoir discussed in \cite{BLTV}, see also Corollary \ref{thm:1ther} below, we cannot expect
$\scrF_t$ and $\widetilde \scrF_t$ to stay close uniformly in $t$. To see this it is enough to look at the steady states
of the two evolutions considered in Theorem \ref{thm:2ther}.

Indeed, on the one hand $\scrF_\infty:=\lim_{t\to\infty}\scrF_t$ is given by the rotational average at fixed total
momentum of $\scrF_0$ on all its variables. More precisely let $\vec e_i$, $i=1,2,3$ be the standard basis in $\R^3$
and, for fixed $P$ consider the vectors $E_i=(\v e_i,\ldots,\v e_i)\in\R^{3P}$. The set $\mathcal O$ of unitary
transformations with determinant one that preserve $E_i$, $i=1,2,3$, is a subgroup of the special  unitary group
$SO(3P)$. Let $d O$ be the normalized Haar measure on $\cal O$. Given a function $F(\ul v)$ from $\R^{3P}$ to $\R$ we
define its fixed momentum rotational average as
\begin{equation}\label{def:rotav}
	\mathcal R[F](\ul v)=\int_{\mathcal O} F(O\ul v)d\sigma(O)\, .
\end{equation}
Clearly $\mathcal R[F](\ul v)$ depends only on $\|\ul v\|$ and $\v V=\sum_{i=1}^N \vec v_i$. Thus, in particular,
$\scrF_\infty$ is invariant under permutation of all of its arguments. On the other hand, using an argument developed in \cite{CLM,Evans} (see also \cite{CELMM18}) we will show in Appendix \ref{app:basic}
that $\tilde f_\infty=\lim_{t\to\infty}\tilde f_t$ exists and is approached exponentially fast so that $\widetilde
\scrF_\infty:=\lim_{t\to\infty}\widetilde\scrF_t=\Gamma_+\tilde f_\infty\Gamma_-$ . Thus $\widetilde \scrF_\infty$ is far
from $\scrF_\infty$, for any $N$. See Section \ref{sec:out} for further discussion of this point.

We can now look back at the situation where there is only one reservoir. A straightforward adaptation of the proof of 
Theorem \ref{thm:2ther} yields the following.

\begin{corollary} \label{thm:1ther} Let $f_0(\underline v)$ be a probability distribution on $\mathds
R^{3M}$ with $E_4(f_0)<\infty$, where $E_4$ is defined as in \eqref{eq:E4}, and let $F_0(\underline v,\underline w)=f_0(\underline
v)\Gamma_T^N(\underline w)$ as discussed in \eqref{eq:ini1}. Consider the two evolved distributions on $\mathds 
R^{3(N+M)}$ given by
\[
F_t=e^{\mathcal L_T}F_0\qquad\hbox{and}\qquad
\widetilde F_t=\Gamma_T^N e^{\widetilde{\mathcal L}_T}f_0
\]
with $\mathcal L_T$ and $\widetilde{\mathcal L_T}$ given by \eqref{eq:gen1res} and \eqref{eq:gen1ther} respectively.
Then we have%\underconstruction
\begin{equation}\label{eq:stima1}
d_2(F_t,\widetilde F_t)\leq 
C\frac{M}{\sqrt
N}E_4(f_0)^{\frac56}(1-e^{-\frac\mu9t})d_2(f_0,\Gamma_T^M)^{\frac16}\, .
\end{equation}
%where $C = C(T)$ is a suitable constant depending only on $T$.
\end{corollary}

As for the 1 dimensional case, we see that $\widetilde F_\infty=\lim_{t\to\infty}\widetilde F_t=\Gamma_T^{N+M}$ while
$F_\infty=\lim_{t\to\infty}F_t=\mathcal R[F_0]$, see \eqref{def:rotav}.
In Appendix \ref{app:SteadyStates} we prove that, as for the model in one dimension, if
$F_0(\uv,\uw)=f_0(\uv)\Gamma^N_T(\uw)$ then
\[
d_2(F_\infty,\Gamma_{T,N+M})\leq d_2(f_0,\Gamma_T^M)\frac{M}{N+M}.
\]
Thus our result in Corollary \ref{thm:1ther} is surely not optimal, as far as the asymptotic behavior in $N$ and $t$ is
concerned. On the other hand, the proof of Lemma \ref{lem:prop5new} makes it quite clear that, at least for short time
$t$, the behavior in $1/\sqrt N$ is not an artifact of our techniques. Such a behavior is a new feature of the three
dimensional evolution, see \cite{BLTV}, and it is our opinion that it is linked to the preservation of the
total moment by the evolution generated by \eqref{eq:gen2res}. See Section \ref{sec:out} for further observations.

\section{A new inequality}\label{newinequ}

The above results are based on an extension of a functional inequality that first appeared in \cite{BLTV}.  The proof
contained in that paper is too complex and contains a gap. Thus we first report a more general version of the original
statement and its simplified proof and then present the extension we will use here. 
 
Consider a $C^3$ function  $H:\mathds R^3\mapsto\mathds C$ with $H(0)=0$ and $\nabla H(0)=0$ and a $C^0$ function 
$g:\mathds R^3\mapsto\mathds C$ such that 
\[
|g(\v\eta)|\leq \frac 1{1+T \|\v \eta\|^2}
\]
for some $T>0$.
For $\ul \eta=(\v \eta_0,\ldots,\v \eta_N)\in\mathds R^{3N}$ let
\[
g^{N}(\ul \eta)=\prod_{i=1}^N g(\v \eta_i)
\] 
and set $\ul \eta^i=(\v\eta_0,\ldots,\v\eta_{i-1},\v\eta_{i+1},\ldots, \v\eta_N)$. Given $\v\xi\in\R^p$, $p\geq1$, we 
want to 
control the {\it interlaced sum}
\begin{equation}\label{eq:DN}
\mathcal D_N(H,\v\xi):=\sup_{\ul \eta\not=0}\frac{1}{\|\ul \eta\|^2+\|\v\xi\|^2}
\left|\sum_{i=1}^N g^{ N - 1}(\ul \eta^i)H(\v\eta_i)\right|
\end{equation}
and show that it can be estimated in terms of %of $\mathcal D_1(H,\v\xi)$ uniformly in $N$ where
\[
\mathcal D_1(H,\v\xi)=\sup_{\v \eta\not=0}\frac{\left|H(\v\eta)\right|}{\|\v \eta\|^2+\|\v\xi\|^2}
\]
uniformly in $N$.
%depends only on $H$ and $\|\vxi\|$. 
More precisely we will prove the following.

\begin{mylemma} \label{lem:prop5} 
Let $H: \mathds R^3\to\mathds C$ be a $C^3$ function such that $H(0)=0$, $\nabla H(0)=0$ and
\begin{equation}\label{eq:H3}
 \|H\|_3:=\max_{\v\beta\in\mathds N^3, \|\v\beta\|_1\leq 3}\bigl\|
\partial^{\v\beta}H\bigr\|_\infty<\infty\, .
\end{equation}
and let $g:\mathds R^3\mapsto \mathds C$ be a $C^0$ function such that
\[
|g(\v\eta)|\leq \frac1{1+T\|\v\eta\|^2}
\]
for some $T>0$. If we call
\[
\mathcal D_N(H,\xi)=\sup_{\ul \eta\not=0}\frac{1}{\|\ul \eta\|^2+\|\v\xi\|^2}
\left|\sum_{i=1}^N g^{ N - 1}(\ul \eta^i)H(\v\eta_i)\right|
\]
then we have 
\begin{equation}\label{eq:estima1}
\mathcal D_N(H,\v\xi)\leq C \|H\|_3^{\frac 23}\mathcal D_1(H,\v\xi)^{\frac13}\, .
\end{equation}
%where $C = C(T)$ is a suitable constant depending only on $T$.
\end{mylemma}
 
 This Lemma is very similar to Proposition 5 in \cite{BLTV} but it contains a few improvements. For one, it
incorporates a more general function $g$ in place of the Fourier transform of the Maxwellian at temperature $T$.
Moreover, it does not require $H$ to be even, which is a necessary improvement for our current setting. This is the
reason for  the cubic root in the main estimate, as opposed to the square root in the analogous estimate from
\cite{BLTV}.

In the following we will need to apply Lemma \ref{lem:prop5} to a function $H:\mathds R^3\mapsto \mathds C$ with 
$H(0)=0$, but $\nabla H(0)\not =0$. We thus first prove the following estimate.

 \begin{mylemma} \label{lem:stimal} 
Let $H: \mathds R^3\times \mathds R^p\to\mathds C$ be a $C^{3,1}$ function such that  $H(0,0)=0$,
$\nabla_{\v\eta}H(0,0)=0$ and, for every $\v\xi\in\R^p$, $H(0,\v\xi)=0$. Assume moreover that
\begin{equation}\label{eq:HDerivBounds1}
\|H\|_{3,0}:=\sup_{\v\xi}\|H(\cdot,\v\xi)\|_3<\infty\,. 
\end{equation}
and
 \begin{equation}\label{eq:HDerivBounds2}
|H|_{1,1}:=\sup_{\v\xi}\bigl\|\nabla_{\v\xi}\nabla_{\v\eta} H(0,\v\xi)\bigr\| < \infty
\end{equation}
Then calling
\begin{equation}\label{eq:defalpha}
\alpha(\v\xi) := \frac{\|\nabla_\eta H(0,\xi)\|}{\|\v\xi\|\sqrt{1 +\|\v\xi\|^2}}
\end{equation}
we get
\begin{equation}
\alpha(\v\xi)\leq 4\min\set{\|H\|_{3,0}^{\frac12}\calD_1(H,\v\xi)^{\frac12},|H|_{1,1}}
\end{equation}
\end{mylemma}

From Lemmas \ref{lem:prop5} and \ref{lem:stimal}, the extension we will need in this paper follows. 
 
 \begin{mylemma} \label{lem:prop5new} 
 Let $H: \mathds R^3\times \mathds R^p\to\mathds C$ be be a $C^{3,1}$ function such that  $H(0,0)=0$,
$\nabla_{\v\eta}H(0,0)=0$ and, for every $\v\xi\in\R^p$, $H(0,\v\xi)=0$. Assume moreover that
\begin{equation}\label{eq:H31}
\|H\|_{3,1}:=\max_{\|\v 
 \beta_\eta\|_1<3,\|\v\beta_\xi\|<1}\bigl\|\partial_{\v\eta}^{\v\beta_\eta}\partial_{\v\xi}^{\v\beta_\xi}
 H\bigr\|_\infty< \infty\, .
\end{equation}
Let also $g:\mathds R^3\mapsto\mathds C$ be a $C^0$ function such that 
\[
|g(\v\eta)|\leq \frac1{1+T\|\v\eta^2\|}
\]
for some $T>0$.
Then we have 
\begin{equation}\label{eq:estima2}
\mathcal D_N(H,\v\xi)\leq  C \sqrt{N}\|H\|_{3,1}^{\frac56}
D_1(H,\v\xi)^{\frac16}\, .
\end{equation}
%
% \begin{equation}\label{eq:estima2}
% \mathcal D_N(H,\v\xi) \leq C_1\sqrt{N}
% \min\set{ \|H\|_{3,0}^{\frac12}D_1(H,\v\xi)^{\frac12},|H|_{1,1}} +
% C_2\|H\|_{3,1}^{\frac56} D_1(H,\v\xi)^{\frac16} .
% \end{equation}
%for suitable constants $C=C(T)$.
 \end{mylemma}
 
% Observe that, if $H$ does not depend on $\xi$, then \eqref{eq:estima2} reduces to \eqref{eq:estima1}. On the other hand
% in the following we will use that under the assumption of Lemma \ref{lem:prop5new} we have
% \begin{equation}\label{eq:Lemmaeasy}
% \mathcal D_N(H,\v\xi)\leq  C(T) \sqrt{N}\|H\|_{3,1}^{\frac56}	D_1(H,\v\xi)^{\frac16}\, .
% \end{equation}
% as it follows easily from \eqref{eq:estima2}.
%
 	 
\section{Outlook}\label{sec:out}

Two limitations of our results appear evident when looking at the statements of Theorem \ref{thm:2ther} and Corollary
\ref{thm:1ther}. 

In the first place, notwithstanding the fact that we consider the system to be `small', we still want to think of the
system as a macroscopic object. That is, we want $M$ to be a fraction of the Avogadro number. In both of our results,
the difference between the system+reservoir(s) evolution and system+thermostat(s) evolution behave as $M/\sqrt{N}$. This
implies that, for the result to be interesting, we need $N>>M^2$, and thus $N$ has to be un-physically large.

Additionally, as we noted after Corollary \ref{thm:1ther}, such behavior, at least for $t$ very large, cannot be
optimal. We think that, by adapting the proof of the main theorem in \cite{Hagop} to a three dimensional evolution, one
can show that the contribution of the order of $M/\sqrt{N}$ in \eqref{eq:stima1} decay exponentially in time leaving an
asymptotic estimate of the order of $M^2/N$. This would represent just a mild improvement with respect to the present
estimate. Moreover the rate of decay, as far as one can get from \cite{Hagop}, would be of the order of $1/N$ and thus
rather too slow to be relevant.

It is quite clear from the proof of Theorem \ref{thm:2ther} that the main difference between the three-dimensional and
the one-dimension Kac evolution is the preservation of the total momentum in addition to the preservation of the total
energy. On the other hand, the metric $d_2$ seems to be well adapted to take advantage of the latter but contains no
reference to the former. We think that it should be possible to modify the $d_2$ metric taking into account the
preservation of momentum and obtaining better estimates. As partial validation of this idea, we observe that if $M=1$
and $f_0$ is rotationally invariant, then a direct computation shows that 
%$\nabla_\veta \widehat G_s(0,\uxi)=0$  for every $\uxi$ and every $s$, see \eqref{eq:HatGDef} below.
we can apply Lemma \ref{lem:prop5} to obtain an estimate of the order of $1/N$.
 
%In particular, the fact that there is an indirect interaction between the reservoirs, eventually brings the combined
%system to equilibrium at the average temperature.  On the other hand, since the reservoirs interact only through the
%system, we can see that the flow of energy between $R_+$ and $R_-$ is initially proportional to $(T_+-T_-)M$. Since
%$M\ll N$, we expect that the state of the reservoirs will not change significantly on a time scale much shorter than
%$N$. This is the situation made precise by Theorem \ref{thm:2ther}.

In the second place, as already observed, Theorem \ref{thm:2ther} is only relevant if $t<<\sqrt N$. To better
understand the long time evolution of the system + 2 reservoir, we can define the temperature, or better average
kinetic energy \red{per degree of freedom}, in the reservoir, $\tau_+(t)$, $\tau_-(t)$, and in the system, $\tau_S(t)$,
as
\[
\begin{aligned}
	\tau_S(t)=\frac1{3M}\int_{\mathds R^{3(2N+M)}}\|\uv\|^2 \scrF_t(\ul u,\ul v,\ul w)d\ul vd\ul u d\ul w\\
	\tau_+(t)=\frac 1{3N}\int_{\mathds R^{3(2N+M)}}\|\uu\|^2 \scrF_t(\ul u,\ul v,\ul w)d\ul vd\ul u d\ul w\\
%	e_-(t)=\frac1N\int_{\mathds R^{3(2N+M)}}\|\uw\|^2 \scrF_t(\ul u,\ul v,\ul w)d\ul vd\ul u d\ul w
\end{aligned}
\]
and similarly for $\tau_-(t)$. It is not hard to show that they evolve according to Newton's Law of cooling:
\begin{equation}\label{eq:Newton}
	\left\{
	\begin{aligned}
		\dot \tau_-(t)&=\frac{\mu M}{3N} (\tau_S(t) -  \tau_-(t))\\
		\dot \tau_S(t)&=\frac\mu3(\tau_-(t) -  \tau_S(t))+\frac\mu3(\tau_+(t) -  \tau_S(t))\\
		\dot \tau_+(t)&=\frac{\mu M}{3N}( \tau_S(t) - \tau_+(t))\, ,
	\end{aligned}
	\right.
\end{equation}
where $\tau_\sigma(0)=T_\sigma$ for $\sigma\in\{+,-\}$ while we set $\tau_S(0)=T_S$.
See Appendix \ref{app:moments} for a derivation. \red{It is easy to solve \eqref{eq:Newton} and obtain
\begin{equation}\label{eq:IVPsol}
	\begin{pmatrix}
		 \tau_+(t) \\ \tau_S(t) \\ \tau_-(t)
	\end{pmatrix}=T_\infty\begin{pmatrix}	1 \\ 1 \\ 1
\end{pmatrix}+
	\frac{T_+-T_-}2\begin{pmatrix}
		1 \\ 0 \\ -1
	\end{pmatrix} e^{-\frac {\mu M}{3N} t}+
	\frac{T_+-2T_S+T_-}{2\overline N}\begin{pmatrix}
		M \\ -2N \\ M
	\end{pmatrix}e^{-\frac \mu3\left(2+\frac MN\right)t}\, .
\end{equation}
where $\overline N=2N+M$ is the total number of particles while $T_\infty=(NT_++MT_S+NT_-)/\overline N$ is the
average temperature.} On a short time scale, $\tau_S(t)$ converges exponentially, with a rate of order 1, to
$T_\infty$. This is the time frame over which our Theorem gives interesting information. On a longer time scale, of
the order of $N$, $\tau_+(t)-\tau_-(t)$ tend to vanishes at a rate proportional to $N^{-1}$. This evolution of
$\tau_\pm(t)$ is the main reason why $\mathcal F(t)$ and $\widetilde{ \mathcal F}$ cannot stay close for long time.

On the other hand, due to the
slow rate of such evolution and the presence of the Kac collisions in the reservoirs, one can heuristically expect that
the state of the reservoirs in uniformly close to $\Gamma^N_{\tau_\pm(t)}$, that is, an equilibrium state at temperature
$\tau_\pm(t)$. Thus for $t$ larger than the initial transient, we will have $\scrF_t\simeq \Gamma^N_{\tau_+(t)}\bar
f_t\Gamma^N_{\tau_-(t)}$ where $\bar f_t$ is the steady state of the system of $M$ particles interacting with two
Maxwellian thermostat at temperature $\tau_+(t)$ and $\tau_-(t)$. This suggests replacing the evolution generated by
$\widetilde \cL$ by considering a system interacting with two reservoirs whose temperatures evolve in time.

\section{Proof of the key Lemmas}\label{sec:proof_prop}

We begin with the proof of Lemma \ref{lem:prop5}. The proof reported here is quite similar to the one in \cite{BLTV} but
it is simpler in several points and it deals with a more general case. Moreover it fills a gap in the previous proof. We will
then show how Lemma \ref{lem:prop5new} follows from Lemma \ref{lem:prop5} with the aide of Lemma \ref{lem:stimal}.

%%%%%%%%%%%%%%%%%%%%%
%%%%%%%%%%%%%%%%%%%%%
%%%%%%%%%%%%%%%%%%%%%
%%%%%%%%%%%%%%%%%%%%%
%%%%%%%%%%%%%%%%%%%%%
%%%%%%%%%%%%%%%%%%%%%

\subsection{Proof of Lemma \ref{lem:prop5}}

As observed in \cite{BLTV}, we can easily show that
\[
\calD_N(H,\vxi)\leq\calD_1(H,0)\, .
\]
On the other hand $\calD_1(H,0)\geq \calD_1(H,\vxi)$ so that this estimate is not enough for our purpose. What prevent us 
to get a better estimate this way is the fact that the supremum over $\ueta$ may be achieved for $\|\ueta\|$ very large 
thus making $\|\vxi\|^2$ essentially negligible in the denominator of \eqref{eq:DN}.

Thus our proof is essentially  divided into two parts. In the first part, we show that $\mathcal D_N(H,\v\xi)$ can be
controlled by $\mathcal D_1(H,0)/(1+T\|\vxi\|^2)$ by suitably restricting the set of $\ul \eta\ne0$ we need to consider
when taking the supremum in \eqref{eq:DN} and then applying a convexity argument. Indeed, the presence of the term
$g^{N-1}(\ueta^i)$ essentially implies that we can restrict our attention to those $\ul \eta$ which lie in a small box
around 0. Convexity of the box and of the function to be estimated inside the box then allows us to explicitly compute
the supremum there. In the second part, we show that $\mathcal D_1(H,0)/(1+T\|\vxi\|^2)$ can be controlled in turn by
$\mathcal D_1(H,\v\xi)$ via a Taylor expansion.

First note that $\calD_N(H,\v\xi)$ only depends on $\xi=\|\vxi\|$ so that we
can assume that $p=1$ and write
\[
\mathcal D_N(H,\xi)=\sup_{\ul \eta\not=0}\frac{1}{\|\ul \eta\|^2+\xi^2}
\left|\sum_{i=1}^N g^{ N - 1}(\ul \eta^i)H(\v\eta_i)\right|
\]
while
\begin{equation}\label{eq:D1}
\calD_1(H,\xi)=\sup_{\eta\not=0}\frac{|H(\eta)|}{\|\v\eta\|^2+\xi^2}
\end{equation}
It follows that, for all $\v\eta \in \R^3$, we have
\[
H(\v\eta) \leq \min\set{\calD_1(H,0)\|\v\eta\|^2, \calD_1(H,\v\xi)(\|\v\eta\|^2 + \xi^2)}.
\]
With this in mind, we define $\widetilde H(\v\eta,\xi) \geq 0$ by
\begin{equation}
\widetilde H(\v\eta,\xi) := \min\set{\calD_1(H,0)\|\v\eta\|^2, \calD_1(H,\v\xi)(\|\v\eta\|^2 + \xi^2)}
\end{equation}
so that $H(\v\eta) \leq \widetilde H(\v\eta,\xi)$ and thus
\begin{equation}\label{eq:DvsDtilde}
	\calD_N(H,\xi) \leq \calD_N(\widetilde H,\xi).
\end{equation} 
It will be more convenient to estimate $\calD_N(\widetilde H,\xi)$ instead of $\calD_N(H,\xi)$ directly.
Observe that for
\[
\eta_0^2 = \frac{\calD_1(H,\xi)\xi^2}{\calD_1(H,0) - \calD_1(H,\xi)}.
\]
we get
\[
\calD_1(H,0) \eta_0^2 = \calD_1(H,\xi) (\eta_0^2 + \xi^2)
\]
so that we can write $\widetilde H$ as
\begin{equation}
\widetilde H(\v\eta,\xi) = \begin{cases} \calD_1(H,0) \|\v\eta\|^2; &\|\v\eta\| \leq \eta_0\\ 
\calD_1(H,\xi)(\|\v\eta\|^2 + \xi^2); &\|\v\eta\|> \eta_0
\end{cases}.
\end{equation}

We now call
\[
G^N(\ul\eta)=\prod_{i=1}^N\frac 1{1+T \|\v\eta_i\|}\, .
\]
Since $\widetilde H$ is positive we can replace $g^N$ with $G^N$ in the definition of $\calD_N(\widetilde H,\xi)$. Thus
from now on we will study
\[
\calD_N(\widetilde H,\xi)=\sup_{\ul \eta\not=0}\frac{1}{\|\ul \eta\|^2+\xi^2}
\sum_{i=1}^N G^{N-1}(\ul \eta^i)\widetilde H(\v\eta_i,\xi)
\]

Our next step is to show that we may restrict our attention to the box $\set{\ul \eta: \|\v\eta_j\|\leq \eta_0}.$ This
will essentially follow from the monotonicity of the two functions $G(\v\eta)$ and $\frac{ka^2 + x^2}{a^2 + x^2}$ and
the fact that $\frac{|a| + |b|}{|c| + |d|} \leq \max\set{\frac{|a|}{|c|}, \frac{|b|}{|d|}}.$

\begin{myclaim}
For $\widetilde H$ as above, we have
\[
\sup_{\ul \eta \ne 0} \frac{\sum_{j = 1}^N \widetilde H(\v\eta_j, \xi) G^{N - 
1}(\ul \eta^j)}{\|\ul \eta\|^2 + \xi^2} \leq 2 \sup_{0<\|\ul\eta\|_\infty \leq \eta_0} \frac{\sum_{j = 1}^N 
\widetilde H(\v\eta_j, \xi) G^{N - 1}(\ul \eta^j)}{\|\ul \eta\|^2 + \xi^2}.
\]
where $\|\ul\eta\|_\infty=\max_{k}\|\v\eta_k\|$.
\end{myclaim}

\begin{proof}
For ease of notation, we will write
\[
\mathcal{H}(\ul\eta,\xi) := \frac{\sum_{j = 1}^N \widetilde H(\v\eta_j, \xi) G^{N - 
1}(\ul\eta^j)}{\|\ul\eta\|^2 + \xi^2}.
\]
First, suppose, for some $k \geq 1,$ $\ul\eta = (\v\eta_1,...,\v\eta_k,0,...,0)$ with $\|\veta_j\| \geq \eta_0$ for all
$j \leq k,$ and denote $\ul\eta_{0,k} = (\v\eta_0,...,\v\eta_0, 0,..., 0),$ for some $\veta_0$ such that
$\|\v\eta_0\|=\eta_0$, and where there are $k$ non-zero components. Using the monotonicity of $G$ and the fact that the
function $\frac{ka^2 + x}{a^2 + x}$ is decreasing in $x$, we get
\begin{equation}
\begin{aligned}
\mathcal{H}(\ul\eta,\xi) &= \frac{\calD_1(H,\xi)\sum_{i = 1}^k (\xi^2 + |\v\eta_i|^2)G^{N - 
1}(\ul\eta^i)}{\xi^2 +\|\ul\eta\|^2}\\ 
&\leq  \frac{\calD_1(H,\xi) G(\vec\eta_0)^{N - 1}(k\xi^2 +
\|\ul\eta\|^2)}{\xi^2 + \|\ul\eta\|^2}\\
&\leq \frac{\calD_1(H,\xi) G(\vec\eta_0)^{N - 1}(k\xi^2 + \|\ul\eta_{0,k}\|^2)}{\xi^2 + 
\|\ul\eta_{0,k}\|^2}\\
&= \frac{\sum_{i = 1}^k  \widetilde H(\v\eta_0,\xi) G^{N - 1}(\ul\eta_{0,k}^i)}{\xi^2 + \|\ul\eta_{0,k}\|^2}.
\end{aligned}
\end{equation}
Hence
\[
\mathcal{H}(\ul\eta,\xi) \leq \mathcal{H}(\ul\eta_{0,k},\xi).
\]
More generally, suppose $\ul\eta= (\v\eta_1,...,\v\eta_N)$ with $\|\v\eta_1\|,...,\|\v\eta_k\| \geq \eta_0$ and
$\|\v\eta_{k +1}\|,...,\|\v\eta_N\| < \eta_0$, for some $k$. Set $\ul\eta_{*} = (\v\eta_1,...,\v\eta_k,0,...,0),$ and
$\ul\eta_{**} = (0,...,0,\v\eta_{k + 1},...,\v\eta_N).$ Then we have
\begin{equation}
\begin{aligned}
\mathcal{H}(\ul\eta,\xi) &= \frac{\sum_{i = 1}^N \widetilde H(\v\eta_i,\xi)G^{N - 1}(\ul\eta^i)}{\xi^2 + \|\ul\eta\|^2}\\
&\leq\frac{\sum_{i = 1}^k \calD_1(H,\xi)(\xi^2 + \|\v\eta_i^2\|)G^{N - 1}(\ul\eta^i) + \sum_{i = k + 1}^N 
\calD_1(H,\xi)\|\v\eta_i\|^2 G^{N - 1}(\ul\eta^i)}
{\frac12\xi^2 +\frac12 (\|\v\eta_1\|^2 + \cdots + \|\v\eta_k^2\|) + 
\frac12\xi^2 +\frac12(\|\v\eta_{k + 1}^2\| + \cdots + \|\v\eta_N^2\|)}\\
&\leq 2\max\set{\frac{\sum_{i = 1}^k \calD_1(H,\xi)(|\xi|^2 + \|\v\eta\|^2)G^{N - 1} 
(\ul\eta^i)}{\xi^2 + (\|\v\eta_1^2\| + \cdots + \|\v\eta_k^2\|)}, 
\frac{\sum_{i = k + 1}^N \calD_1(H,0) \|\v\eta_i\|^2 G^{N - 
1}(\vec\eta^i)}{\xi^2 + (\|\v\eta_{k + 1}\|^2 + \cdots + \|\v\eta_N\|^2)}}.
\end{aligned}
\end{equation}
Since $\eta_*$ is of the form we considered in the first case, for the first term we get
\[
\frac{\sum_{i = 1}^k \calD_1(H,\xi)(\xi^2 + \|\v\eta_i\|^2)G^{N - 1}
(\ul\eta^i)}{\xi^2 + (\|\v\eta_1\|^2 + \cdots + \|\v\eta_k^2)\|} \leq 
\mathcal{H}(\ul\eta_{*},\xi) \leq \mathcal{H}(\ul\eta_{0,k},\xi).
\]
where we have used again the monotonicity of $G$ for the first inequality. On the other hand, the second term satisfies
\[
\frac{\sum_{i = k + 1}^N \calD_1(H,0) \|\v\eta_i\|^2 G^{N - 1}(\ul\eta^i)}
{\xi^2 + (\|\v\eta_{k + 1}\|^2  + \cdots + \|\v\eta_N\|^2)} \leq\mathcal H(\ul\eta_{**}, \xi).
\]
Since $\vec\eta_{**}$ has all component of norm less than $\eta_0$, combining these two cases yields the claim.
\end{proof}

We are now left with
\begin{equation}\label{eq:box}
\sup_{\ul\eta \ne 0} \frac{\sum_{j = 1}^N \widetilde H(\v\eta_j, \xi) G^{N - 1}(\ul\eta^j)}{\|\ul\eta\|^2 + \xi^2} \leq 
2 \sup_{\vec\eta \ne 0, \|\v\eta_j\| \leq \eta_0} \frac{\sum_{j = 1}^N \calD_1(H,0) \|\v\eta_j\|^2 
G^{N - 1}(\ul\eta^j)}{\|\ul\eta\|^2 + \xi^2}.
\end{equation}
Next, using a convexity argument, we will essentially reduce the supremum over the set $\set{\v\eta_j\,|\,\|\v\eta_j\| 
\leq \eta_0}$ to te supremum over an interval.

\begin{myclaim}
The supremum on the right-hand side of \eqref{eq:box} is achieved when at most one $\v\eta_i$ satisfies $0 < \|\v\eta_i\| 
< \eta_0.$
\end{myclaim}

\begin{proof}
Since the right hand side of \eqref{eq:box} depends only on the $\|\v\eta_i\|^2$ we can set $x_i=\|\v\eta_i\|^2$. Thus 
we call $B_r := \set{\ul x: \sum_{i=1}^N x_i = r, 0\leq x_i \leq \eta_0^2 }$ and observe that for  $x_i>0$ we have
\[
\prod_{i=1}^N \frac 1{1+T x_i}\leq \frac 1{1+ T\|\ul x\|_1}
\] 
where $\|\ul x||_1=\sum_{i=1}^N |x_i|$. We thus get
\[
 \sup_{\vec\eta \ne 0, \|\v\eta_j\| \leq \eta_0} \frac{\sum_{j = 1}^N \|\v\eta_j\|^2 \hat 
G^{N - 1}(\ul\eta^j)}{\|\ul\eta\|^2 + \xi^2}\leq\sup_{0<r\leq N\eta_0^2}\sup_{\ul x\in B_r}
\frac 1{1+T\|\ul x\|_1} \frac{\sum_{i=1}^N x_i(1+T x_i)}{\|\ul x\|_1+ \xi^2}
\]
Since $B_r$ is a convex set and $\sum_i x_i (1+T x_i)$ is a convex function we get
\[
\sup_{\ul x \in B_r}
\frac 1{1+T\|\ul x\|} \frac{\sum_{i=1}^N x_i(1+T x_i)}{\|\ul x\|_1+ \xi^2}\leq
\sup_{\ul x \in \partial B_r}
\frac 1{1+T\|\ul x\|} \frac{\sum_{i=1}^N x_i(1+T x_i)}{\|\ul x\|_1+ \xi^2}
\]
Observe now that $\partial B_r$ contains only points for which at most one $0<x_i<\eta_0^2$. The claim now follows.
\end{proof}

Thus we have 
\[
\begin{aligned}
&\sup_{\ul\eta \ne 0, \|\v\eta_j\| \leq \eta_0} \frac{\sum_{j = 1}^N \widetilde H(\v\eta_j, \xi) G^{N -  1}(\ul\eta^j)}
{| \ul\eta|^2 + \xi^2}  \\
&\qquad\leq \calD_1(H,0)\max_{0 \leq k \leq N - 1} \sup_{0 \leq |\eta| \leq \eta_0} 
\frac{k\eta_0^2(1+T((k - 1)\eta_0^2 + 
\eta^2))^{-1} + \eta^2(1+T k \eta_0^2)^{-1}}{k\eta_0^2 + \eta^2 + \xi^2}
\end{aligned}
\]
Observe now that
\[
\begin{aligned}
&\frac{k\eta_0^2(1+T((k - 1)\eta_0^2 + 
\eta^2))^{-1} + \eta^2(1+T k \eta_0^2)^{-1}}{k\eta_0^2 + \eta^2 + \xi^2}
\leq\\
&\qquad\max\set{\frac{\eta_0^2}{\eta_0^2+\xi^2/2}, 
\frac{(k-1)\eta_0^2(1+T((k - 1)\eta_0^2 + 
\eta^2))^{-1} + \eta^2(1+T k \eta_0^2)^{-1}}{(k-1)\eta_0^2 + \eta^2 + \xi^2/2}}\, .
\end{aligned}
\]
Moreover we have
\[
\calD_1(H,0)\frac{\eta_0^2}{\eta_0^2+\xi^2/2}\leq 2\calD_1(H,0)\frac{\eta_0^2}{\eta_0^2+\xi^2}=2\calD_1(H,\xi)
\]
where the last identity follows from the definition of $\eta_0$, while
\[
\frac{(k-1)\eta_0^2(1+T((k - 1)\eta_0^2 + 
\eta^2))^{-1} + \eta^2(1+T k \eta_0^2)^{-1}}{(k-1)\eta_0^2 + \eta^2 + \xi^2/2}
\leq 
\frac{((k-1)\eta_0^2+\eta^2)(1+T((k - 1)\eta_0^2 + 
\eta^2))^{-1} }{k\eta_0^2 + \eta^2 + \xi^2/2}
\]
Since the right hand side of the last inequality depends only on $(k-1)\eta_0^2+\eta^2$ we have
\[
\max_{0 \leq k \leq N - 1} \sup_{0 \leq |\eta| \leq \eta_0} 
\frac{((k-1)\eta_0^2+\eta^2)(1+T((k - 1)\eta_0^2 + 
\eta^2))^{-1} }{k\eta_0^2 + \eta^2 + \xi^2/2}
\leq 
2\sup_{y>0}\frac{y(1+T y)^{-1}}{y+\xi^2}
\leq \frac{2}{1+T\xi^2}.
\]
Collecting everything we get
\[
\calD(H,\xi)\leq4\max\set{\calD_1(H,\xi),\frac{\calD_1(H,0)}{1+T\xi^2}}
\]

We now need to control $\calD_1(H,0)/(1+T|\xi^2|)$ with $\calD_1(H,\xi)$. This will be a consequence of the following Claim. 

\begin{myclaim}\label{claim:hr}
Under the assumption of Lemma \ref{lem:prop5} we have
\[
\calD_1(H,\xi)\geq \frac 1{57}\frac{\calD_1(H,0)^3}{\calD_1(H,0)^2+\|H\|_3^2\xi^2}\, .
\]
\end{myclaim}

\begin{proof}
We consider the function $h:\mathds R^+\mapsto\mathds R$ defined as
\[
h(r)=\sup_{\|\omega\|=1}|H(\omega r)|.
\]
and observe that
\[
\calD_1(H,\xi)=\sup_{r>0}\frac{h(r)}{r^2+\xi^2}
\]
Since we can write
\[
H(\v\eta)=\sum_{i,j=1}^3\partial_i\partial_jH(0)\eta_i\eta_j+\sum_{i,j,k=1}^3C_{i,j,k}(\v\eta)\eta_i\eta_j\eta_k
\]
calling $\bar H$ the largest modulus of the eigenvalues of Hessian matrix of $H$ at $\v\eta=0$, we get
\[
|H(\v\eta)|\leq\bar H\|\v\eta\|^2+\|H\|_3\|\v\eta\|^3
\]
while 
\[
\sup_{\omega}|H(r\omega)|\geq r^2\sup_{\omega}\left|\sum_{i,j=1}^3\partial_i\partial_jH(0)\omega_i\omega_j\right|-
r^3\sup_{\omega}\left|\sum_{i,j,k=1}^3C_{i,j,k}(r\omega)\omega_i\omega_j\omega_k\right| 
\]
so that we get
\begin{equation}\label{eq:expa}
\bar H r^2 - \|H\|_3 r^3\leq h(r)\leq \bar H r^2+ \|H\|_3 r^3\, .
\end{equation}
In particular we obtain that $\calD_1(H,0)\geq \bar H$. Moreover from the left hand side of \eqref{eq:expa} we have
\[
\frac{h(r)}{r^2+\xi^2}\geq \frac{\bar H r^2-\|H\|_3r^3}{r^2+\xi^2}.
\]
Taking $r=\frac{\calD_1(H,0)}{3\|H\|_3}$ we get
\begin{equation}\label{eq:left}
\calD_1(H,\xi)\geq \frac{\calD_1(H,0)^2}{27}\frac{3\bar H -\calD_1(H,0)}{\calD_1(H,0)^2+\|H\|_3^2\xi^2}\, .
\end{equation}
On the other hand, the right hand side of \eqref{eq:expa} gives
\begin{equation}\label{eq:right}
h(r)-\calD_1(H,0)r^2\leq (\bar H-\calD_1(H,0)) r^2+ \|H\|_3 r^3
\end{equation}
Let $r_*$ be the smallest $r$ such that 
\[
\calD_1(H,0)) =\frac{h(r_*)}{r_*^2}\, .
\]
From \eqref{eq:right} we see that either $r_*=0$ or $r_*\geq (\bar H-\calD_1(H,0))/\|H\|_3$. Observe that, if $r_*=0$, then
$\calD_1(H,0)=\bar H$ and \eqref{eq:left} gives the thesis. Thus we can assume that $r_*\geq(\bar H-\calD_1(H,0))/\|H\|_3$.
Observing that $r^2/(r^2+\xi^2)$ is increasing in $r$, it follows that
\[
\calD_1(H,\xi)\geq \frac{h(r_*)}{r_*^2+\xi^2}= \calD_1(H,0)\frac{r_*^2}{r_*^2+\xi^2}\geq \frac{\calD_1(H,0)(\bar 
H-\calD_1(H,0))^2}{(\bar 
H-\calD_1(H,0))^2+\|H\|_3^2\xi^2}\geq \frac{\calD_1(H,0)^2(\calD_1(H,0)-2\bar H)}{\calD_1(H,0)^2+\|H\|_3^2\xi^2}
\]
Thus we get
\begin{equation}\label{eq:max}
\calD_1(H,\xi)\geq\frac{\calD_1(H,0)^2}{\calD_1(H,0)^2+\|H\|_3^2\xi^2}\max\set{\frac{(3\bar 
H-\calD_1(H,0))}{27},\calD_1(H,0)-2\bar H}\, .
\end{equation}
Observe finally that if $\alpha_1,\alpha_2>0$ and $\alpha_1+\alpha_2=1$ then $\max\set {a_1, a_2}\geq \alpha_1 
a_1+\alpha_2 a_2$. Taking $\alpha_1=18/19$ and $\alpha_2=1/19$ in \eqref{eq:max} gives
\[
\calD_1(H,\xi)\geq\frac 1{57}\frac{\calD_1(H,0)^3}{\calD_1(H,0)^2+\|H\|_3^2\xi^2}\, .
\]
\end{proof}

To conclude the proof of Lemma \ref{lem:prop5} we observe that $h(r)\leq \|H\|_3 r^2$ so that $\calD_1(H,\xi)\leq
\|H\|_3$. In particular, $\calD_1(H,0) \leq \|H\|_3.$ Combining with Claim \ref{claim:hr} it follows that 
\[
\frac{\calD_1(H,0)}{1+\xi^2/3}\leq 4\|H\|_3^{\frac23} \calD_1(H,\xi)^{\frac13}
\]
so that
\[
\frac{\calD_1(H,0)}{1+T\xi^2}\leq 4\max\set{1,(3T)^{-1}}\|H\|_3^{\frac23} \calD_1(H,\xi)^{\frac13}\, .
\]
Finally we obtain the thesis observing that we can write $\calD_1(H,\xi)\leq \|H\|_3^{\frac23} \calD_1(H,\xi)^{\frac13}$.

%%%%%%%%%%%%%%%%%%%%%
%%%%%%%%%%%%%%%%%%%%%
%%%%%%%%%%%%%%%%%%%%%
%%%%%%%%%%%%%%%%%%%%%
%%%%%%%%%%%%%%%%%%%%%
%%%%%%%%%%%%%%%%%%%%%

\subsection{Proof of Lemma \ref{lem:stimal} and \ref{lem:prop5new}}

We begin proving Lemma \ref{lem:prop5new}. We observe that if $\nabla_{\veta}H(0,\vxi)\not =0$ then, there exists
$\veta_*$ with $\|\veta_*\|=1$ such that, for $\rho$ small enough we have
\[
|H(\rho\veta_*,\vxi)G(\rho\veta_*)|\geq \frac 12|\rho|\|\nabla_{\veta}H(0,\vxi)\|.
\]
For $N$ large enough we can chose $\ueta=\frac 1{\sqrt N}(\veta_*,\ldots,\veta_*)$ and obtain
\[
\calD_N(H,\vxi)\geq \frac {\sqrt N}2 \frac{\|\nabla_{\veta}H(0,\vxi)\|}{1+\|\vxi\|^2}.
\]
This shows that the second term in \eqref{eq:estima2} is not an artifact of our proof.

As for the proof of Lemma \ref{lem:prop5} we can assume that $g(\v\eta)=G(\v\eta)=(1+T\|\v\eta\|^2)^{-1}$. From
\eqref{eq:HDerivBounds1} and \eqref{eq:HDerivBounds2}, we can write \label{sub:proof prop5new} %
\begin{equation}\label{eq:TaylorExpansion}
H(\vec \eta,\vxi) = G(\vec\eta) (\vec \eta \cdot \nabla_{\vec \eta} H(0,\vxi)) +  L(\vec \eta,\vxi).
\end{equation}

Considering that $G(0)=1$ and $\nabla G(0)=0$, we get $L(0,\vxi) = 0$, $\nabla_{\vec \eta} L(0,\vxi) =  0$ and
\[
\|L\|_3\leq \|H\|_{3,0}+G_3\|\nabla H\|_\infty\leq 2G_3  \|H\|_{3,0}
\]
where 
\[
G_3=\max\set{T^{-1},T^2}\sum_{i=0}^3\sup_x\left|\frac{d^i}{dx^i}\frac x{1+x^2}\right|.
\]
It follows that
\begin{align}\label{eq:FirstTriangleIn}
\begin{split}
\mathcal D_N(H,\vxi) &=\sup_{\underline\eta\ne0}\left|
\frac{\sum_{i = 1}^N \left(\vec \eta_i \cdot \nabla_{\vec \eta} H(0,\vxi)G_N(\underline\eta) +  L(\vec
\eta_i,\vxi)G_{N-1}(\underline\eta^i)\right)}{\|\ueta|^2 + \|\vxi\|^2}\right|\\ 
&\leq
\sup_{\underline\eta\ne0}\left| \frac{\sum_{i = 1}^N \vec \eta_i \cdot \nabla_{\vec \eta}
H(0,\vxi)G_{N}(\underline\eta) }{\|\ueta\|^2 + \|\vxi|^2}\right| + \sup_{\underline\eta\ne0}\left| \frac{\sum_{i
= 1}^N  L(\vec \eta_i,\vxi)G_{N-1}(\underline\eta^i)}{\|\ueta\|^2 + \|\vxi\|^2}\right|\\ 
&=
\sup_{\ueta\ne0}\left| \frac{\sum_{i = 1}^N \vec \eta_i \cdot \nabla_{\vec \eta} H(0,\vxi)G_N(\underline\eta)
}{\|\ueta\|^2 + \|\vxi\|^2}\right| + \mathcal D_N( L,\vxi).
\end{split}
\end{align}
	
We can use the Cauchy-Schwartz inequality to obtain
\begin{align}\label{eq:FirstHetaBound}
\begin{split}
\sup_{\underline\eta\ne0} \left| \frac{\sum_{i = 1}^N \vec \eta_i \cdot \nabla_{\vec \eta} H(0,\xi)G_N(\underline\eta) 
}{\|\ueta\|^2 + \|\vxi\|^2}\right| &\leq \sqrt{N} \left|\nabla_{\vec\eta} H(0,\xi)\right| \sup_{\vec x\ne 0} 
\frac{|\vec x| 
G(\vec x) }{\|\vec x\|^2 + \|\vxi\|^2}
\end{split}
\end{align}
	
Observe that we can apply Lemma \ref{lem:prop5} to $\mathcal D_N( L,\vxi)$, notwithstanding the dependence of $L$ on 
$\vxi$. Using also \eqref{eq:TaylorExpansion}, we obtain
\begin{equation}\label{eq:tildeLBound}
\begin{aligned}
\mathcal D_N( L,\vxi) &\leq C \|L\|_3^{\frac23} \mathcal D_1( L,\vxi)^{\frac 13}\\
&\leq C \|L\|_3^{\frac23} \left(\mathcal D_1(H,\vxi)^{\frac13} + \mathcal D_1(\vec\eta \cdot \nabla_{\vec \eta} H(0,\vxi) 
G(\vec \eta),\vxi)^{\frac13}\right)\\
&\leq C G_3^{\frac23} \|H\|_{3,0}^{\frac23} \left(\mathcal D_1(H,\vxi)^{\frac13} + |\nabla_{\vec\eta} H(0,\vxi)|^{\frac13} 
\left(\sup_{\vec x\ne 0}\frac{|\vec x|G(\vec x)}{\|\vec x\|^2 + \|\vxi\|^2}\right)^{\frac 13}\right)\\
\end{aligned}
\end{equation}
Since
\[
\sup_{\vec x\ne 0}\frac{\|\vec x\|G(\vec x)}{\|\vec x\|^2 + \|\vxi\|^2} \leq \frac{1}{\|\vxi\|\sqrt{1 + \|\vxi\|^2}},
\]
we can combine this with \eqref{eq:FirstTriangleIn}, \eqref{eq:tildeLBound}, and \eqref{eq:FirstHetaBound} to obtain
\begin{equation}\label{eq:ReductionToNabla}
\mathcal D_N(H,\vxi) \leq \sqrt{N}  \frac{\|\nabla_\eta H(0,\vxi)\|}{\|\vxi\|\sqrt{1 + \|\vxi\|^2}} + C
G_3^{\frac23}\|H\|_{3,0}^{\frac23} \left(\mathcal D_1(H,\vxi)^{\frac13} + \left(\frac{\|\nabla_{\vec\eta}
H(0,\vxi)\|}{\|\vxi\|\sqrt{1 + \|\vxi\|^2}}\right)^{\frac 13}\right).
\end{equation}

\begin{proof}[Proof of Lemma \ref{lem:stimal}]
Since $H(0,\vxi)=0$ for every $\vxi$ we see that $\nabla_{\vxi}H(0,0)=0$ so that 
\begin{equation}\label{eq:H11}
\alpha(0)<\infty
\end{equation}
and $\calD(H,0)\leq\infty$. Analogously to  \eqref{eq:TaylorExpansion}, we can write $H(\vec \eta,\xi)
=  \vec \eta \cdot \nabla_{\vec \eta} H(0,\xi) +  \widetilde L(\vec \eta,\xi)$ so that we get
\begin{align*}
|H(\vec \eta,\vxi)| &\geq |\vec \eta \cdot \nabla_{\vec \eta} H(0,\vxi)| - |\widetilde L(\vec \eta,\vxi)|\\
&\geq |\vec \eta \cdot \nabla_{\vec \eta} H(0,\vxi)| - \|\vec \eta\|^2\norm{H}_{3,0}\\
&= \|\vxi\| \frac{|\vec \eta \cdot \nabla_{\vec \eta} H(0,\vxi)|}{\|\vxi\|} - \|\vec \eta\|^2\norm{H}_{3,0}.
\end{align*}
Thus
\begin{align*}
\mathcal D_1(H,\vxi) &= \sup_{\vec \eta \ne 0}\frac{|H(\vec \eta,\vxi)|}{\|\vec \eta\|^2 +\|\vxi\|^2}\\ &\geq \sup_{\vec
\eta\ne 0}\frac{\|\vxi\| \frac{|\vec \eta \cdot \nabla_{\vec \eta} H(0,\vxi)|}{\|\vxi\|} - \|\vec 
\eta\|^2\norm{H}_{3,0}}{\|\veta\|^2 + \|\vxi\|^2}
\end{align*}
Since the supremum over $\vec \eta\ne0$ is certainly larger than the value at any particular $\vec \eta,$ setting $\vec 
\eta_0$ to be the vector parallel to $\nabla_{\vec\eta} H(0,\xi)$ with norm $\|\vec \eta_0\| = 
\frac{\|\vxi\|\alpha(\vxi)}{2\norm{H}_{3,0}}$ yields
\begin{align*}
\mathcal D_1(H,\vxi) &\geq \frac{\|\vxi\|^2
\frac{\alpha(\vxi)^2}{4\norm{H}_{3,0}}}{\|\vxi\|^2\frac{\alpha(\vxi)^2}{4\norm{H}_{3,0}^2} + \|\vxi\|^2}\\ &=
\frac{\alpha(\xi)^2\norm{H}_{3,0}}{\alpha(\xi)^2 + 4\norm{H}_{3,0}^2}
\end{align*}
It follows that
\[
\alpha(\xi)^2 \leq
\frac{4\|H\|_{3,0}^2\calD_1(H,\xi)}{\|H\|_{3,0}-\calD_1(H,\xi)}\leq
8\|H\|_{3,0}\calD_1(H,\xi)
\]
where we used that $\|H\|_{3,0}\geq 2\calD_1(H,\xi)$.	Together with
\eqref{eq:H11}, this gives the thesis.

\end{proof}

Combining \eqref{eq:ReductionToNabla} and Lemma \ref{lem:stimal}, we get
\[
\mathcal D_N(H,\xi) \leq
C\sqrt{N}
\|H\|_{3,0}^{\frac56}\min\set{D_1(H,\xi)^{\frac16},|H|_{1,1}^{\frac16}} +
C\|H\|_3^{\frac23} 
D_1(H,\xi)^{\frac13} .
\]

%%%%%%%%%%%%%%%%%%%%%
%%%%%%%%%%%%%%%%%%%%%
%%%%%%%%%%%%%%%%%%%%%
%%%%%%%%%%%%%%%%%%%%%
%%%%%%%%%%%%%%%%%%%%%
%%%%%%%%%%%%%%%%%%%%%

\section{Proof of Theorem \ref{thm:2ther}}\label{sec:2ther}

In this section we prove Theorem \ref{thm:2ther}. The proof is based on time interpolation between $\scrF_t$ and
$\widetilde{\scrF}_t$ based  on Duhamel's formula and  on Lemma \ref{lem:prop5new}. But first we start reminding the basic
definition from the introduction.

\subsection{The setting}

We consider the evolution generated by the operator
\begin{equation}\label{eq:L}
\cL = \cL_S + \cL_{I_+} + \cL_{I_-} + \cL_{R_+} + \cL_{R_-}% = Q- \Lambda I
\end{equation}
acting on the probability distributions on $\R^{3(M + 2N )},$ where
\begin{equation}
\begin{aligned}
\cL_S &:= \frac{1}{M - 1} \sum_{1 \leq i < j \leq M} (R_{i,j}^S - \mathrm{Id})\\
%=: Q_S - \Lambda_SI
\cL_{R_\sigma} := \frac{\lambda}{N - 1} \sum_{1 \leq i < j \leq N}&
(R_{i,j}^{R_\sigma} - \Id) \qquad\qquad
%=: Q_{R_\sigma} - \Lambda_RI\\
\cL_{I_\sigma} := \frac{\mu}{N} \sum_{i = 1}^N \sum_{j = 1}^{M}
(R_{i,j}^{I_\sigma} - \Id)% =: Q_{I_\sigma} - \Lambda_I I.
\end{aligned}
\end{equation}
where $\sigma=\pm$ and $R_{i,j}^\alpha$, $\alpha\in\set{S,R_+,R_-,I_+,I_-}$ is
the collision operator \eqref{eq:defR}
acting on particles in the system $\alpha=S$, in the reservoirs $\alpha=R_+$, and $R_-$ or one in each as in the
interaction terms with $\alpha=I_+$, and $I_-$.

Moreover we assume that the evolution starts from the initial state
\[
\scrF_0(\uu,\uv, \uw)=\Gamma_+(\uu)f_0(\uv)\Gamma_-(\uw)
\]
where $\Gamma_\sigma$ is the Maxwellian distribution on $N$ particles at temperature 
$T_\sigma$. We thus set
\[
\scrF_t=e^{\cL t}\scrF_0\, .
\]

We want to compare $\scrF_t$ with the distribution
\begin{equation}\label{eq:Ft}
\widetilde \scrF_t=\Gamma_+(\uu)\tilde f_t(\uv)\Gamma_-(\uw)
\end{equation}
where
\[
\tilde f_t=e^{\widetilde \cL_t}f_0
\]
and
\begin{equation}\label{eq:tildeL}
\widetilde \cL = \cL_S + \cL_{B^+} + \cL_{B^-} %=\widetilde Q - \widetilde\Lambda I
\end{equation}
acts on $\R^{3M},$ with
\begin{equation}
\cL_{B^\sigma} := \mu \sum_{i = 1}^M (B^\sigma_i - I)% =: Q_{B^\sigma} - \Lambda_B I
\end{equation}
where $B^\sigma_i$ is the operator \eqref{eq:defB} describing the interaction between particle $i$ and the thermostat at
temperature $T_\sigma$. Although $\widetilde \cL$ act on distribution on $\R^{3M}$, we will also consider its extension
to the distributions on $\R^{3(M+2N)}$ that acts as the identity on the extra variables. With a slight abuse of notation
we will use the same symbol for such an extension thus writing
\[
\widetilde \scrF_t=e^{\widetilde \cL t}\scrF_0=\Gamma_+(\uu)e^{\widetilde \cL t} f_0(\uv)\Gamma_-(\uw)
\]

%%%%%%%%%%%%%%%%%%%%%
%%%%%%%%%%%%%%%%%%%%%
%%%%%%%%%%%%%%%%%%%%%
%%%%%%%%%%%%%%%%%%%%%
%%%%%%%%%%%%%%%%%%%%%
%%%%%%%%%%%%%%%%%%%%%

\subsection{Estimating $d_2(\widetilde \scrF_t, \scrF_t)$}\label{subs:Duhamel}

Our first step in estimating the difference between $\scrF_t$ and $\widetilde\scrF_t$ is based on the Duhamel formula
\[
e^{ \cL t} - e^{\widetilde \cL t} = \int_0^t e^{\cL (t - s)} (\cL - \widetilde \cL)e^{\widetilde \cL s} ds.
\]
that yields
\begin{align}
d_2(e^{\cL t} \scrF_0, e^{\widetilde \cL t}\scrF_0)
&\leq \int_0^t d_2(e^{\cL(t - s)} \cL e^{\widetilde \cL s} \scrF_0, e^{ \cL (t - s)} \widetilde \cL e^{\widetilde \cL s} \scrF_0) ds\leqq
\int_0^t d_2( \cL \widetilde \scrF_s, \widetilde \cL \widetilde \scrF_s) ds.
\end{align}
where $\widetilde \scrF_s = e^{\widetilde \cL s}\scrF_0$
and we have used Corollary \ref{Cor:GapDecay}. 

From \eqref{eq:Ft}, using the definitions \eqref{eq:L} and \eqref{eq:tildeL} we get
\begin{equation}\label{eq:DuhamelTriangle}
d_2(e^{ \cL t}\scrF_0, e^{\widetilde \cL t}\scrF_0) \leq \int_0^t d_2(\cL_{B_+}\widetilde \scrF_s, \cL_{I_+} \widetilde \scrF_s) ds 
+  \int_0^t d_2(\cL_{B_-}\widetilde \scrF_s, \cL_{I_-} \widetilde \scrF_s) ds.
\end{equation}
We will control the first integral on the right-hand side. The second term can be controlled in the same way.

The first step to estimating this expression is understanding the Fourier transform of each term. From now on, we will
let $\uxi, \ueta,$ and $\uzeta$ denote the dual Fourier variables of $\uv, \uu,$ and $\uw$ respectively. We have
\begin{equation}\label{eq:FirstReservoirD2}
	d_2(\cL_{B_+}\widetilde \scrF_s, \cL_{I_+} \widetilde \scrF_s) = \sup
\frac{\mu}{N}\sum_{i = 1}^{N}\sum_{j = 1}^M \frac{
	\widehat{R_{i,j}^{I_+}\scrF_s} - \widehat{B^+_j \scrF_s}}{|\uxi|^2 +
|\ueta|^2 + |\uzeta|^2}
\end{equation} 
Observe that if $f(\vv_1,\vv_2)$ is a distribution on $\R^6$ then
\[
\widehat{R[f]}(\vxi_1,\vxi_2)=\int_{\S^2} \hat f(\vxi_1^*(\omega),\vxi_2^*(\omega))d\omega=R[\hat f](\vxi_1,\vxi_2)
\]
while if $g(\vv)$ is a distribution on $\R^3$ then
\[
\widehat{B[g]}(\vxi)=\int_{\S^2} \hat f(\vxi^*(\omega))\widehat\Gamma_T(0^*(\omega))d\omega=R[\hat g\widehat\Gamma_T](\vxi,0)
\]
where $0^*(\omega)=(\xi\cdot\omega)\omega$.
Moreover for a general distribution of the form $\scrF(\vec u, \vec v, \vec w) = \Gamma^{N}_-(\uu)  f(\uv) 
\Gamma^{N}_-(\uw)$, we have
\begin{align}
\label{eq:FirstFourierFormula}
\widehat{R_{i,j}^{I_+} \scrF}(\ueta, \uxi, \uzeta) &= R_{i,j}^{I_+} \left[\hat f
\widehat\Gamma^1_+\right] (\veta_i,\uxi)\widehat\Gamma^{N - 1}_+(\ueta^i)
\widehat\Gamma^{N}_-(\uzeta)\\
\widehat{B_{j}^+ \scrF}(\ueta, \uxi, \uzeta) &= \widehat{R_{i,j}^{I_+} F}(
\ueta, \uxi, \uzeta)\bigr|_{\veta_i=0}
\widehat\Gamma^{1}_+(\veta_i) =R_{i,j}^{I_+} \left[\hat f
\widehat\Gamma^1_+\right]
(0,\uxi)\widehat\Gamma^{1}_+(\veta_i)\widehat\Gamma^{N - 1}_+(\ueta^i)
\widehat\Gamma^{N}_-(\uzeta)
\end{align}

For ease of notation, we will not indicate the dependence on $\uzeta$ if not
necessary. For $1 \leq i \leq N,$ define
\begin{equation}\label{eq:HatGDef}
 \widehat G_s(\veta_i, \uxi) =
 \sum_{j = 1}^M \left(R^{I_+}_{i,j}[\hat f_{s}\widehat\Gamma_+](\veta_i, \uxi)
- R^{I_+}_{i,j}[\hat f_{s}\widehat\Gamma_+] (0,
\uxi)\widehat\Gamma_+(\veta_i)\right).
\end{equation}
Now \eqref{eq:FirstReservoirD2} becomes
\begin{equation}\label{eq:Mterms}
d_2(\cL_{B_+} \widetilde \scrF_s, \cL_{I_+} \widetilde \scrF_s) = 
\frac{\mu}{N} \sup_{\uxi,\ueta,\uzeta\not=0}\left| \frac{\sum_{i = 1}^{N}
\widehat G_{s}(\veta_i, \uxi)
\widehat\Gamma^{N-1}_+(\ueta^i)\widehat\Gamma^{N}_+(\uzeta)}
{|\uxi|^2 + |\ueta|^2 + |\uzeta|^2} \right|
\leq 
\frac{\mu}{N} \sup_{\uxi,\ueta\not=0}\left| \frac{\sum_{i = 1}^{N}  \widehat
G_{s}(\veta_i, \uxi) \widehat\Gamma^{N-1}_+(\ueta^i)}
{|\uxi|^2 + |\ueta|^2 } \right|
\end{equation}
With the aim of applying Lemma \ref{lem:prop5new} to the above expression we observe that
\begin{enumerate}
\item For any $\uxi \in \R^{3M}$, $\widehat G_s$ has a zero of order 1 at $(0,\uxi).$ That is, $\widehat G_s(0, 
\vec \xi) = 0$;

\item $\widehat G_s$ has a zero of order 2 at $(0, 0).$ That is,
$\nabla_{\veta_i} \widehat G_s(0, 0) = 0$ and clearly
$\nabla_{\uxi} \widehat G_s(0, 0) = 0$;

\item In Appendix \ref{app:moments} we show that there exists a constant $C$ such that for every $s$ we have
$\|\widehat G_s\|_{3,1}\leq M C E_4(f_0)$, see Lemma \ref{lem:Gs};
		
\item It is easy to see that $\widehat\Gamma_+(\veta)\leq\frac1{1+T_+\|\veta\|^2}$.
\end{enumerate}

Thus, we see that Lemma \ref{lem:prop5new} is applicable with $H =
\widehat G_s,$ and we obtain
\begin{equation}\label{eq:finfty}	
\begin{aligned}
d_2(\cL_{B_+}\widetilde \scrF_s, \cL_{I_+} \widetilde \scrF_s) &\leq
\frac{\mu}{\sqrt{N}}C M^{\frac56} E_4(f_0)^{\frac56}
d_2(G_s,0)^{\frac16}\\
&\leq \frac{\mu M}{\sqrt{N}}C E_ 4(f_0)^{\frac56}d_2\left(R_{i1}^{I_+} [\tilde
f_s \Gamma^1_+], B_i^+ [\tilde
f_s \Gamma^1_+]\right)^{\frac16}.
\end{aligned}
\end{equation}
where we used the monotonicity of $\Gamma_+$ and the fact that $\tilde f_s$ is permutation invariant so that the 
$M$ term in \eqref{eq:Mterms} gives the same contribution. 

We can now apply Lemma \ref{Lemma:FirstReduction} and Corollary \ref{Cor:GapDecay} to obtain
\begin{equation}\label{eq:Gmfinfty}
	\begin{aligned}
			d_2\left(R [\tilde f_s\Gamma_+], B^+[ \tilde f_s
]\Gamma_+\right) &
          \leq  d_2\left(R [(\tilde f_s-\tilde f_\infty)\Gamma_+],
B^+[ \tilde f_s -\tilde f_\infty]\Gamma_+\right) +
        d_2\left(R [\tilde f_\infty\Gamma_+], B^+[ \tilde
f_\infty ]\Gamma_+\right) \\
			&\leq 2 e^{-\frac{2\mu}3 s}d_2\left(f_0, \tilde f_\infty\right) 
				+ 2 d_2\left(\tilde f_\infty, \Gamma^M_+\right).
	\end{aligned}
\end{equation}
so that it follows that
\[
\begin{aligned}
	\int_0^t d_2(\cL_{B_+}\widetilde \scrF_s,& \cL_{I_+} \widetilde \scrF_s) ds \\
				&\leq 
					2\frac{M\mu C}{\sqrt{N}}E_4(f_0)^{\frac56}
		\int_0^t 	\left(e^{-\frac{2\mu}3 s}d_2(f_0, \tilde f_\infty) 
					+  d_2(\tilde f_\infty, \Gamma_+^M)\right) ^{\frac16}\ ds\\
				&\leq 
					18\frac{M C}{\sqrt{N}}E_4(f_0)^{\frac56}
						 (1 - e^{-\frac\mu9 t}) d_2(f_0, \tilde f_\infty)^{\frac16} 
					+ 2\frac{M\mu C}{\sqrt{N}}E_4(f_0)^{\frac56}
						 d_2(\tilde f_\infty, \Gamma_+^M) ^{\frac16}t\\
\end{aligned}
\]
We clearly obtain an analogous estimate for the second term in \eqref{eq:DuhamelTriangle} so we conclude that
\begin{equation}\label{eq:finalest}
\begin{aligned}
	d_2(e^{ \cL t}\scrF_0, e^{\widetilde \cL t}\scrF_0) 
				&\leq 
					\frac{MC}{\sqrt{N}}E_4(f_0)^{\frac56}
							(1 - e^{-\frac\mu9 t}) d_2(f_0, \tilde f_\infty)^{\frac16} \\
				&+ \frac{M\mu C}{\sqrt{N}}E_4(f_0)^{\frac56}\left(d_2(\tilde
f_\infty, \Gamma_+^M)^{\frac16} +
d_2(\tilde f_\infty, \Gamma_-^M)^{\frac16}\right)t
\end{aligned}
\end{equation}
We observe that, if $T_+ = T_-=T$ then we are essentially in the 1-reservoir situation. \red{In this case
$\Gamma_+=\Gamma_-=\Gamma_T$ so that $\tilde f_\infty=\Gamma_T^M$ and \eqref{eq:finalest} simplifies} to
\begin{align}
	d_2(e^{ \cL t}\scrF_0, e^{\widetilde \cL t}\scrF_0) \leq
\frac{MC}{\sqrt{N}}E_4(f_0)^{\frac56}
							(1 - e^{-\frac\mu9 t}) d_2(f_0, \Gamma_T^M)^{\frac16}
\end{align}
which is the same as the estimate we would obtain by doing the computation directly starting from a 1-reservoir 
set-up. This proves Corollary \ref{thm:1ther}.

In Appendix \ref{app:basic} we show that 
\begin{equation}\label{eq:diffT}
d_2(\tilde f_\infty, \Gamma_\sigma^M)\leq \frac{T_+-T_-}2
\end{equation}
so that the second term in \eqref{eq:finalest} essentially captures the difference between the temperature of the two 
Maxwellian thermostats. This concludes the proof of Theorem \ref{thm:2ther}.

\section*{Delarations}

\subsection*{Conflict of Interest}
The authors have no conflicts of interest to disclose.

\subsection*{Author Contributions}
Federico Bonetto: Formal analysis (equal); Writing – original draft (equal); Writing – review \& editing (equal). 
Michael Loss: Formal analysis (equal); Writing – original draft (equal); Writing – review \& editing (equal). 
Matthew Powell: Formal analysis (equal); Writing – original draft (equal); Writing – review \& editing (equal).

\subsection*{Data Availability}
Data sharing is not applicable to this article as no new data were created or analyzed in this study.

%%%%%%%%%%%%%%%%%%%%%
%%%%%%%%%%%%%%%%%%%%%
%%%%%%%%%%%%%%%%%%%%%
%%%%%%%%%%%%%%%%%%%%%
%%%%%%%%%%%%%%%%%%%%%
%%%%%%%%%%%%%%%%%%%%%

\bibliographystyle{plain}

%\bibliography{nonequi}

%%%%%%%%%%%%%%%%%%%%%
%%%%%%%%%%%%%%%%%%%%%
%%%%%%%%%%%%%%%%%%%%%
%%%%%%%%%%%%%%%%%%%%%
%%%%%%%%%%%%%%%%%%%%%
%%%%%%%%%%%%%%%%%%%%%

\appendix

%%%%%%%%%%%%%%%%%%%%%
%%%%%%%%%%%%%%%%%%%%%
%%%%%%%%%%%%%%%%%%%%%
%%%%%%%%%%%%%%%%%%%%%
%%%%%%%%%%%%%%%%%%%%%
%%%%%%%%%%%%%%%%%%%%%

\section{Basic Properties of the evolutions}\label{app:basic}

In this appendix we deduce the existence of a steady state for both evolutions and derive \eqref{eq:diffT}. In the following, we will restrict our attention to the metric space $(X, d_2)$ where $X = \{f \in L^1: d_2(f, \Gamma) < \infty\}$.

We start by studying the behavior of the GTW $d_2$ metric under the action of $R$ and $B$ defined in
\eqref{eq:defR} and \eqref{eq:defB} respectively.

\begin{mylemma}\label{Lemma:FirstReduction}
Suppose $f(\vv_1,\vv_2)$ and $g(\vv_1,\vv_2)$ are distributions on $\R^6$ with zero first moment and finite second 
moment. Then
\[
d_2(R[f], R[g]) \leq d_2(f, g)
\]
while if $f(\vv)$ and $g(\vv)$ are distribution on $\R^3$ with zero first moment and finite second 
moment then
\[
d_2(B[f], B[g])\leq  \frac23 d_2(f, g)
\]
\end{mylemma}

	\begin{proof}
	Observe that
		\begin{align}
			d_2\left(R[f], R[g]\right)
			&= \sup_{\vxi_1, \vxi_2} \left|\frac{\int_{\mathbb S^2} (\hat f(\vxi_1^*(\omega)\,\vxi_2^*(\omega)) - 
			\hat g(\vxi_1^*(\omega),\vxi_2^*(\omega)))\
d\omega}{\|\vec\xi_1\|^2 + \|\vec\xi_2\|^2} \right| \\
			&\leq \sup _{\vxi_1, \vxi_2} \int_{\mathbb S^2} 
					\frac{\left|\hat f(\vxi_1^*(\omega), \vxi_2^*(\omega)) -\hat  g(\vxi_1^*(\omega),\vxi_2^*(\omega))\right|}
						{\|\vxi_1^*(\omega)\|^2 +
\|\vxi_2^*(\omega)\|^2}d\omega\leq d_2\left(f, g\right).
		\end{align}
Similarly we get 
\begin{equation}\label{eq:estimateB}
		\begin{aligned}
				d_2(B[f], B[g]) =& \sup_{\vxi}\frac1{\|\vxi\|^2} \left|\int_{\mathbb S^2} \left(\hat f(\vxi-(\omega\cdot \vxi)\omega) - 
			\hat g(\vxi-(\omega\cdot \vxi)\omega)\right)\hat \Gamma_T((\omega\cdot \vxi)\omega) d\omega\right|\\
				&\leq d_2(f, g) \sup_{\vxi}\left(1 - \int
\frac{|(\vec\xi\cdot\omega)|^2}{\|\vec\xi\|^2} d\omega\right)= \frac 23 d_2(f,
g)
		\end{aligned}
\end{equation}
where we have used if $\sigma\in\mathds S^2$ then
\begin{equation}\label{eq:n2}
n_2:=\int_{\mathds S^2}(\sigma\cdot\omega)^2d\omega=\frac12\int_{-\frac\pi2}^{\frac\pi2}\sin^2\theta\cos\theta 
d\theta=\frac13
\end{equation}
\end{proof}
	
Observe now that we can write
\[
\cL_S = \frac{1}{M - 1} \sum_{1 \leq i < j \leq M} (R_{i,j}^S - \Id)=: Q_S -
\Lambda_S\Id
\]
and similar decomposition for $\cL_{R_\sigma}=:Q_{R_\sigma} - \Lambda_R\Id$, $\cL_{I_\sigma}=:Q_{I_\sigma} - 
\Lambda_I\Id$, and $\cL_{B_\sigma}=:Q_{B_\sigma} - \Lambda_B\Id$. We also set
\[
\begin{aligned}
Q&=Q_S + Q_{I_+} + Q_{I_-} + Q_{R_+} + Q_{R_-} &\qquad\qquad
\Lambda&=\Lambda_S+2\Lambda_R+2\Lambda_I \\
\widetilde Q&=Q_S + Q_{B^+} + Q_{B^-} &\qquad\qquad
\widetilde\Lambda&=\Lambda_S+2\Lambda_B \,,
\end{aligned}
\]
so that we get
\[
\cL=Q-\Lambda\;\Id\qquad\qquad \widetilde\cL=\widetilde Q-\widetilde\Lambda\;\Id \, .
\]

It will also be useful to have the following corollary, which informally states that the Kac evolution is not expanding while 
the thermostated evolution is a contracting.
\begin{corollary}\label{Cor:GapDecay}
Let $\scrF$ and $\scrG$ be two distributions on $\R^{3(M+2N)}$ with 0 mean and finite second moment, then
\begin{equation}\label{eq:nonexpa}
d_2(e^{ \cL t}\scrF, e^{ \cL t} \scrG) \leq d_2(\scrF, \scrG).
\end{equation}
Let now $f$ and $g$ be two distributions on $\R^{M}$ with 0 mean and finite second moment, then
	\begin{equation}\label{eq:contra}
		d_2(e^{ \widetilde\cL t}f, e^{ \widetilde\cL t} g) \leq e^{-\frac{2\mu}3 t}d_2(f, g)\,
	\end{equation}
\end{corollary}

\begin{proof}
	We begin by writing
		\begin{equation}\label{eq:powerseries}
			e^{ \cL t} = e^{-\Lambda t} \sum_{n = 0}^\infty \frac{t^n}{n!}Q^n,
		\end{equation}
From Lemma \ref{Lemma:FirstReduction}, it is easy to see that
\[		
d_2(Q_\alpha \scrF, Q_\alpha \scrG) \leq \Lambda_\alpha d_2(\scrF, \scrG); \quad \alpha\in \set{S, I_\sigma, R_\sigma}\, ,
\]
so that
\[
d_2(Q \scrF, Q \scrG)\leq \Lambda  d_2(\scrF, \scrG)\, .
\]
Thus
\[
d_2(e^{ \cL t}\scrF, e^{ \cL t} \scrG) \leq e^{-\Lambda t} \sum_{n = 0}^\infty \frac{t^n}{n!}d_2(Q^n\scrF,Q^n\scrG)\leq 
d_2(\scrF, \scrG)e^{-\Lambda t} \sum_{n = 0}^\infty \frac{t^n}{n!}\Lambda^n 
\]
where we have used the convexity of $d_2$.

For \eqref{eq:contra} we observe that, reasoning like for \eqref{eq:estimateB} we get 
\begin{equation}\label{eq:QBf}
d_2(Q_{B_\sigma} f, Q_{B_\sigma} g)
\leq \mu d_2(f, g) \sup_{\uxi\not=0}\left(M - \sum_{i=1}^M\int \frac{|(\vec\xi_i\cdot\omega)|^2}{|\uxi|^2} d\omega\right)= 
\left (\Lambda_B-\frac\mu3 \right) d_2(f, g)
\end{equation}
%where we have used that 
%\[
%\int \frac{|(\vec\xi_i\cdot\omega)|^2}{|\uxi|^2} d\omega=\frac{\|\vxi_i\|^2}{\|\uxi\|^2}\int  
%\frac{|(\vec\xi_i\cdot\omega)|^2}{|\vxi_i|^2} d\omega=\frac{c\|\vxi_i\|^2}{\|\uxi\|^2}\, .
%\]
The second inequality now follows as before. 
\end{proof}

An immediate consequence of Corollary \ref{Cor:GapDecay} is that the thermostated evolution has a unique steady-state
which depends only on the temperatures of the thermostats. 

In order to derive \eqref{eq:diffT}, it will be helpful to establish the following claim, which follows from a simple computation.

\begin{mylemma} \label{lem:diffT}
The $d_2$ distance between two Maxwellian distribution is proportional to their temperature difference, that is
\[
d_2(\Gamma^1_{T_+},\Gamma^1_{T_-})=\frac{T_+-T_-}2\, .
\]
\end{mylemma}

\begin{proof}
Observe that if $T_+>T_-$, using that $1-e^{-x}\leq x$, we get
\[
d_2(\Gamma^1_{T_+},\Gamma^1_{T_-})=\sup_{x>0}\frac{e^{-T_-x/2}-e^{-T_+x/2}}x=
\sup_{x>0}\frac{e^{-T_-x/2}(1-e^{-(T_+-T_-)x/2})}x\leq 
\sup_{x>0}\frac{|T_+-T_-|}2 e^{-|T_+-T_-|x/2} 
\]
while the opposite inequality is evident.
\end{proof}

We can now derive \eqref{eq:diffT}.

\begin{proof}[Derivation of \eqref{eq:diffT}]
We look at the case $\sigma=+$ since the other is very similar. Observe that by definition
	\[
	(Q_S + Q_{B_+} + Q_{B_-})\tilde f_\infty = \widetilde\Lambda \tilde f_\infty
	\]
	while
	\[
	(Q_S + 2Q_{B_+}) \Gamma_+^M = \widetilde\Lambda \Gamma_+^M.
	\]
	Hence 
	\[
	\begin{aligned}
		\widetilde \Lambda d_2(\tilde f_\infty, \Gamma_+^M) &= d_2((Q_S +
Q_{B_+} + Q_{B_-})\tilde f_\infty, (Q_S + 2Q_{B_+}) \Gamma_+^M)\\
		&\leq d_2(Q_S\tilde f_\infty, Q_S\Gamma_+^M) + d_2((Q_{B_+} + Q_{B_-})\tilde f_\infty, 2Q_{B_+} \Gamma_+^M) \\
		&\leq \Lambda_S d_2(\tilde f_\infty, \Gamma_+^M) + d_2((Q_{B_+} + Q_{B_-})\tilde f_\infty, (Q_{B_+} + Q_{B_-})\Gamma_+^M) + 
		d_2( Q_{B_-}\Gamma_+^M, Q_{B_+} \Gamma_+^M) \\
%		&\leq \Lambda_S d_2(\tilde f_\infty, \Gamma_+^M) + 2\left(\Lambda_B-\frac\mu3\right) d_2(\tilde f_\infty, \Gamma_\sigma^M) + 2\Lambda_B d_2(B_{T_+}\Gamma_{T_+} , B_{T_-}\Gamma_{T_+}).
	\end{aligned}
	\]
Observe now that, reasoning like for \eqref{eq:estimateB} and \eqref{eq:QBf} we get
\[
d_2( Q_{B_-}\Gamma_+^M, Q_{B_+} \Gamma_+^M)\leq \mu d_2(\Gamma_+,\Gamma_-)\sup_{\vxi\not=0}\left(\sum_{i=1}^M\int \frac{|(\vec\xi_i\cdot\omega)|^2}{|\uxi|^2} d\omega\right)\leq \frac\mu 3 d_2(\Gamma_+,\Gamma_-)
\]
so that, using \eqref{eq:QBf}, we obtain
\[
\widetilde \Lambda d_2(\tilde f_\infty, \Gamma_+^M)\leq \Lambda_S d_2(\tilde f_\infty, \Gamma_+^M) + 2\left(\Lambda_B-\frac\mu3\right) d_2(\tilde f_\infty, \Gamma_+^M) + \frac\mu 3 d_2(\Gamma_+,\Gamma_-)
\]
	The conclusion thus follows from Lemma \ref{lem:diffT} above and recalling that $\Lambda_S + 2\Lambda_B =\widetilde\Lambda.$
\end{proof}

%%%%%%%%%%%%%%%%%%%%%
%%%%%%%%%%%%%%%%%%%%%
%%%%%%%%%%%%%%%%%%%%%
%%%%%%%%%%%%%%%%%%%%%
%%%%%%%%%%%%%%%%%%%%%
%%%%%%%%%%%%%%%%%%%%%

\section{Moments of $f_0$ and derivatives of $\widehat G_s$.}\label{app:moments}

 In this appendix we first look at the evolution of the moments of the distribution $\tilde f_t$. Although it is quite
natural to expect such moments to be uniformly bounded in $t$, it is less easy to analyze their behavior in $M$. We
will proof here that $E_4(\tilde f_t)$ is bounded uniformly in $t$ and $M$ in term of $E_4(f_0)$. The proof is based on
a direct computation coupled with a dimensional argument. It extends the analogous proof given in \cite{BLTV}, for a
much simpler case. Since it turns out to be rather notationally involved, we will not report it in full details. We
will then derive the Newton Law of cooling \eqref{eq:Newton}.

Continuing the definitions given just before Theorem \ref{thm:2ther}, we call $V^k_o$ the space of homogeneous
polynomials $p_k(\uv)=\sum_{\mathbf i\in I^k}p_{\mathbf i}\,\uv_{\mathbf i}$ of degree $k$ and $V^l$ the space of
polynomials of degree $l$, that is $p\in V^l$ if $p(\uv)=\sum_{k=0}^l p_k(\uv)$ with $p_k\in V^k_o$ where $V_o^0=\mathds
R$. On $V^l$ we consider the scalar product
\begin{equation}\label{eq:prodP}
(p,q)=\sum_{\mathbf i\in \tilde I^l}p_{\mathbf i}q_{\mathbf i}
\end{equation}
where $\tilde I^l=\bigcup_{k=0}^l I^k$.

To a polynomial $p$ we associate its $f_0$ expectation
\[
\EXP_0(p):=\int_{\R^{3M}}p(\uv)f_0(\uv)d\uv=\sum_{\mathbf i\in \tilde I^l}m_{\mathbf i}(f_0)\,p_{\mathbf i},
\]
where $m_{\mathbf i}(f)=\int_{\R^{3M}}v_{\mathbf i}f(\uv)d\uv$ is the $\mathbf i$-th moment of $f_0$. We can thus write
\[
\int_{\R^{3M}}p(\uv) \cL_S[f](\uv)d\uv=:\int_{\R^{3M}} L_S[p](\uv) f(\uv)d\uv
\]
where $L_S$ is the linear operator acting on $V^l$ associated to $\mathcal L_S$. Observe that $L_SV^k_o\subset V^k_o$.  
We can now write
\begin{equation}
\begin{aligned}
\EXP_t(p):=\int_{\R^{3M}}p(\uv)f_t(\uv)d\uv=&\int_{\R^{3M}}p(\uv)e^{\cL_S t}[f_0](\uv)d\uv=\\
&\int_{\R^{3M}}e^{ L_S t}[p](\uv) f_0(\uv)d\uv=
\sum_{\mathbf i\in \tilde I^l}m_{\mathbf i}\left(e^{ L_S t}[p]\right)_{\mathbf i}=(e^{ L_S t}[p],\underline m)
\end{aligned}
\end{equation}
where $\underline m=\{m_{\mathbf i}: \mathbf i\in \tilde I^l\}$ and, for simplicity sake, we set $m_{\mathbf 
i}=m_{\mathbf i}(f_0)$.\footnote{In the following, with a slight abuse of notation, we will use $p$ to indicate both the
polynomial $P(\uv)$ and the sets of its coefficients $p=\{p_{\mathbf i}: \mathbf i\in \tilde I^l\}$.}

In a similar way we can define $L_{B}$ as the linear operator acting on $V^l$ associated to $\mathcal L_B$. Observe that
$L_{B}V^k_o\subset V^k$ and that, calling $L_B^{k,l}:=P_{V^l}L_B\bigl|_{V^k_o}$, where $P_V$ is the orthogonal projector
on the subspace $V$, we have $L_B^{k,l}\not =0$ if and only if $k-l$ is non negative and even, due to the parity of the
Maxwellian distribution. The following Lemma collect basic properties of these operators.

\begin{mylemma} \label{lem:Lself}
$L_S$ is self adjoint, $(L_S[p],p)\leq 0$. Moreover $L_B^{k,k}$ is self-adjoint, $(L_B^{k,k}[p],p)\leq 0$,
$(L_B^{k,k}[p],p)< 0$ if $k>0$. Finally $\|L_B^{k,l}\|<CM$.
\end{mylemma}
\begin{proof}
Given $i,i'\in \{0,\ldots,M\}$ and $\omega\in \mathds S_2$, let $r_{i,i'}(\omega)$ be the linear transformation on
$\R^{3M}$ described in \eqref{eq:romega}. A direct computation shows that
\[
L_S\bigr|_{V^k_0}=\frac{\lambda_S}{M-1}\sum_{i<j} \int_{\mathds S^2}( r_{i,j}(\omega)^{\otimes k}-\mathrm{Id})d\omega\, 
.
\]
The statement follows immediately from the fact that $r_{i,j}(\omega)$ is unitary and self-adjoint.

Let now $\tilde r_i(\omega)$ be the linear transformation on $\R^{3M+3}$ described in \eqref{eq:rbomega}. Then $\tilde 
r_i(\omega)^{\otimes k}$ acts on homogeneous the polynomials of degree $k$ in $3M+3$ variables. Let $\widetilde P$ be 
the projector between the space of homogeneous polynomials in $3M+3$ variables to the space of homogeneous polynomials 
in $3M$ variables. Then we have
\[
L_B^{k,k}=\mu \sum_i \int_{\mathds S^2}( \widetilde P\tilde r_{i}(\omega)^{\otimes k}\widetilde P-\mathrm{Id})d\omega
\] 
and the thesis follows from the fact that $\widetilde P\tilde r_{i}(\omega)^{\otimes k}\widetilde P$ is a principal
minor of a self-adjoint unitary matrix. Finally we observe that also $L_B^{k,l}$ is the sum of $M$ terms that can be
written in term of suitable (non-orthogonal) projections of $\tilde r_{i}(\omega)^{\otimes k}$.
\end{proof}

Thus Lemma \ref{lem:Lself} tells us that $\widetilde L= L_S +  L_{B_+} +  L_{B_-}$ is block triangular with negative
definite diagonal blocks but potentially large off-diagonal blocks. It follows that, if $q^2(\uv):=\frac{1}{M}\sum_i
\|\vv_i\|^{4}$ then $\EXP_t(q^2)\leq CM^{2}\|q^2\|\|E_4(f_t)\|$. Observe moreover that $\|\underline m\|=O(M^2) E_4(f_0)$
while $\|q^2\|=M^{-\frac 12}$. It follows that $\EXP_t(q^2)$ is bounded uniformly in $t$ but not in $M$.

To obtain a bound uniform in $M$ we need to look at the structure of $L_S$ and $L_B$. To avoid overburdening the
notation we will only consider the case $l=4$. Moreover we will restrict our attention on the subspace of
permutation invariant polynomials containing only monomials of even degree. They clearly form an $L$ invariant subspace.

We observe that the polynomials
\[
\begin{aligned}
&p_{\bf j}^{4,1}(\uv)=\frac{1}{M}\sum_{i} v_{i,j_1}v_{i,j_2}v_{i,j_3}v_{i,j_4}\\ 
&p_{\bf j}^{4,2}(\uv)=\frac{1}{\binom M 2}\sum_{i_1<i_2}  v_{i_1,j_1}v_{i_1,j_2}v_{i_2,j_3}v_{i_2,j_4}\quad
p_{\bf j}^{4,2'}(\uv)=\frac{1}{\binom M 2}\sum_{i_1<i_2}  v_{i_1,j_1}v_{i_1,j_2}v_{i_1,j_3}v_{i_2,j_4}\\
&p_{\bf j}^{4,3}(\uv)=\frac{1}{\binom M 3}\sum_{i_1<i_2<i_3}  v_{i_1,j_1}v_{i_1,j_2}v_{i_2,j_3}v_{i_3,j_4}\quad
p_{\bf j}^{4,4}(\uv)=\frac{1}{\binom M 4}\sum_{i_1 < i_2 < i_3 < i_4}  v_{i_1,j_1}v_{i_2,j_2}v_{i_3,j_3}v_{i_4,j_4}
\end{aligned}
\]
with $\mathbf j\in J^4:=\set{1,2,3}^4$, form an orthogonal basis in the subspace of permutation invariant polynomials in
$V^4_o$. Let $V^{4,1}_o=\mathrm{span}\set{p_{\bf j}^{4,1}\, :\,\mathbf j\in J^4}$ the
subspace of {\it one particle} polynomials, $V^{4,2}_o=\mathrm{span}\set{p_{\bf j}^{4,2}, p_{\bf j}^{4,2}\, :\, \mathbf
	j\in J^4}$ the subspace of {\it two particles} polynomials, and similarly for $V^{4,3}_o$ and $V^{4,4}_o$. 
	In a
similar way we consider the polynomials
\[
p_{\bf j}^{2,1}(\uv)=\frac{1}{Z_{2,1}}\sum_{i} v_{i,j_1}v_{i,j_2}\qquad 
p_{\bf j}^{2,2}(\uv)=\frac{1}{Z_{2,2}}\sum_{i_1<i_2} v_{i_1,j_1}v_{i_2,j_2} 
\]
with $\mathbf j\in\set{1,2,3}^2$, and the subspaces $V^{2,1}_o$ and $V^{2,2}_o$.
\red{
\begin{mylemma}\label{lem:lemref}
For every $\mathbf j\in J^4$ we have
\begin{equation}\label{eq:lemref1}
(e^{t L}p_{\mathbf j}^{4,k},\underline m)\leq CE_4(f_0)
\end{equation}
while for every $\mathbf j\in J^2$ we get
\begin{equation}\label{eq:lemref2}
(e^{t L}p_{\mathbf j}^{2,k},\underline m)\leq CE_4(f_0)\, .
\end{equation}
Finally
\begin{equation}\label{eq:lemref3}
(e^{t L}1,\underline m)\leq CE_4(f_0)\, .
\end{equation}
\end{mylemma}}

\begin{proof} A direct computation shows that
\[
\Bigl\|P_{V^{l,n}_o}L_S\bigr|_{V^{l,m}_o}\Bigr\|\leq C M^{-\frac{|m-n|}2}
\]
if $|m-n|\leq 1$ and 0 otherwise while $\Bigl\|P_{V^{l,n}_o}L_{B_{\pm}}\bigr|_{V^{l,m}_o}\Bigr\|\leq C \delta_{m,n}$. 
Finally 
$P_{V^{l,n}_o}L\bigr|_{V^{l,m}_o}$ is negative definite if $l>0$.  It follows
from a perturbative argument that
\[
\left\|P_{V^{l,n}_o}e^{L t}\bigr|_{V^{l,m}_o}\right\|=CM^{-\frac{|m-n|}2}e^{-\kappa t}
\]
for a suitable $\kappa>0$. Observing that $\|p_{\mathbf j}^{4,k}\|\leq C M^{-\frac k2}$ while $\|P_{V^{4,k}_o}
\underline m\|\leq C M^{\frac k2}E_4(f_0)$, we get that
\begin{equation}\label{eq:dimen}
(P_{V^4_o}e^{t L}p_{\mathbf j}^{4,k},\underline m)\leq \sum_{l=0}^4 \Bigl\|P_{V^{4,l}_o}\underline
m\Bigr\|\left\|P_{V^{4,l}_o}e^{L t}\bigr|_{V^{4,k}_o}\right\|\left\|p_{\mathbf j}^{4,k}\right\|\leq Ce^{-\kappa 
t}E_4(f_0)\, .
\end{equation}
Since $L_{B_\pm}$ are block triangular we can write
\[
P_{V^2_o}e^{t L}p_{\mathbf j}^{4,k}=\int_0^t P_{V^2_o}e^{s L}(L_{B_+}^{4,2}+L_{B_-}^{4,2})P_{V^4_o}e^{s L}p_{\mathbf 
j}^{4,k}dt
\] 
Another direct computation show that $\Bigl\|P_{V^{2,m}_o}L_B\bigr|_{V^{4,n}_o}\Bigr\|\leq C M^{\frac{n-m}2}$ for 
$n-m=0,1$ and 0 otherwise. Repeating the argument in \eqref{eq:dimen} we get that also $(P_{V^2_o}e^{t 
L}p_{\mathbf j}^{4,k},\underline m)\leq Ce^{-\kappa t}E_4(f_0)$. Finally we write
\[
P_{V^0}e^{t L}p_{\mathbf j}^{4,k}=\int_0^t P_{V^0}e^{s L}(L_{B_+}^{4,0}+L_{B_-}^{4,0})P_{V^4_o}e^{s L}p_{\mathbf 
j}^{4,k}dt+
\int_0^t P_{V^0}e^{s L}(L_{B_+}^{2,0}+L_{2_-}^{2,0})P_{V^2_o}e^{s L}p_{\mathbf j}^{4,k}dt
\]
that gives $(P_{V^0}e^{t L}p_{\mathbf j}^{4,k},\underline m)\leq C$ since $P_{V^0}L=0$. Summing up we get \eqref{eq:lemref1}. A similar argument gives \eqref{eq:lemref2} and \eqref{eq:lemref3}.

\end{proof}

The following Lemma is a direct consequence of the above discussion.

\begin{mylemma} \label{Lemma:MomentBounds} 
Suppose $f_0(\uv)$ is a permutation invariant distribution on $\R^{3M}$ such that 
$E_4(f_0)< \infty$. Then for every $t>0$ we have 
\[	
E_4(\tilde f_t) \leq C E_4(f_0)\, .
\]
\end{mylemma}

\begin{proof}
For every $\mathbf i=((i_1,j_1),\ldots(i_4,j_4))\in I^k$, \red{using H\"older inequality with suitable exponents} we 
can write
\[
e_{\mathbf i}(\tilde f_t)\leq \left(\prod_{l=1}^k\int_{\R^{3M}}(v_{i_l,j_l})^4 \tilde f_t(\uv)d\uv\right)^{\frac 14}\, .
\]
Since  $\int_{\R^{3M}}\tilde f_t(\uv)d\uv=1$ and thus $E_4(\tilde f_t)\geq 1$, applying Lemma \ref{lem:lemref} we get
\[
e_{\mathbf i}(\tilde f_t)\leq 
CE_4(f_0)\,.
\]
The thesis follows immediately.
\end{proof}

We can now formulate the main result of this Appendix in the following Lemma.

\begin{mylemma}\label{lem:Gs}
 Let $\widehat G_s$ be as defined in \eqref{eq:HatGDef}, that is,
 \begin{equation}
 \widehat G_s(\veta_i, \uxi) =
 \sum_{j = 1}^M \left(R^{I_+}_{i,j}[\hat f_{s}\widehat\Gamma_+](\veta_i, \uxi)
- R^{I_+}_{i,j}[\hat f_{s}\widehat\Gamma_+] (0,
\uxi)\widehat\Gamma_+(\veta_i)\right).
\end{equation}
Then we have
 \[
  \|\widehat G_s\|_{3,1}\leq CM E_4(f_0)\, .
 \]

\end{mylemma}

\begin{proof}
It is enough to observe that
\[
\|R^{I_+}_{i,j}[\hat f_{s}\widehat\Gamma_+]\|_{3,1}\leq C\max_{\v\beta\in\mathds N^3, \|\v\beta\|_1\leq 4}\bigl\|
\partial^{\v\beta}\hat f_s\bigr\|_\infty\leq CE_4(\tilde f_s)\, .
\]
	
% Since 
% \[
% G_s(\vw, \uv) = \sum_{j = 1}^M \left(R_j[\tilde f_s\Gamma_T](\vw,\uv)-\Gamma_T(\vw)B_j[\tilde f_s](\uv) \right),
% \]
% we have, from Jensen,
% \begin{align}
% \|\widehat G_s\|_{3,1}\leq \sum_{j = 1}^M \widetilde E_4\left(R_j[\tilde 
%f_s\Gamma_T](\vw,\uv)-\Gamma_T(\vw)B_j[\tilde f_s](\uv) \right).
% \end{align}
% By Lemma \ref{Lemma:Hbound} we have
% \[
% \widetilde E_4\left(R_j[f\Gamma_T](\vw,\uv)-\Gamma_T(\vw)B_j[f](\uv) \right) \leq  C(E_4(\tilde f_s) + 1).
% \]
% From Lemma \ref{Lemma:MomentBounds}, we have
% \[
% E_4(\tilde f_s) \leq C (E_4(f_0) + 1).
% \]
% Together, this yields 
% \[
% \|\widehat G_s\|_{3,1} \leq CM (E_4(f_0) + 1).
% \]
\end{proof}

We can now derive \eqref{eq:Newton}.

\begin{proof}[Derivation of \eqref{eq:Newton}] For $f:\mathds R^6\mapsto \mathds R$ we have
\[
\begin{aligned}
	\int\|\vv\|^2(R-I)[f](\vv,\vw)d\vv d\vw &=
	\int \left[\|\uv-((\uv-\vw)\cdot\omega)\omega\|^2-\|\vv\|^2\right]
	f(\vv,\vw)d\omega d\vw d\uv\\
	&= \frac 13\int ( \|\vw\|^2- \|\vv\|^2)f(\vv,\vw) d\vw d\uv \, .
\end{aligned}
\]
we thus get
	\[
\begin{aligned}
	\int\|\uv\|^2 (R_{i,j}^{R_\sigma}-I)[\scrF_t](\ul u,\ul v,\ul w)d\ul vd\ul u d\ul w&=0\\
	\int\|\uv\|^2 (R_{i,j}^{S}-I)[\scrF_t](\ul u,\ul v,\ul w)d\ul vd\ul u d\ul w&=0\\
	\int\|\uv\|^2 (R_{i,j}^{I_+}-I)[\scrF_t](\ul u,\ul v,\ul w)d\ul vd\ul u d\ul w&=
	\frac13\int(\|\vu_i\|^2-\|\vv_j\|^2) \scrF_t(\ul u,\ul v,\ul w)d\ul vd\ul u d\ul w
\end{aligned}
\]
and similarly for $R_{i,j}^{I_-}-I$. Summing over $i$ and $j$ we get the equation for $e_S(t)$. The others two
equations follow in a similar manner.
\end{proof}

\section{Rotational average at fixed momentum}
\label{app:SteadyStates}

For $N$ and $k<N$ consider a distribution $f:\R^{3k}\mapsto \R$ such that
\[
\int_{\R^{3k}}f(\ul v)d\ul v=\bar f \qquad\mathrm{and}\int_{\R^{3k}}\ul vf(\ul v)d\ul v=0
\]
and its ``extension" $f\Gamma^{N-k}_T$ to a distribution on $\R^{3k}$.
\begin{mylemma}\label{lem:rota}
We have
\[
d_2\left(\mathcal R\bigl[f\Gamma^{N-k}_T\bigr],\bar f\Gamma^N_T\right)\leq \frac kN d_2(f,\bar f\Gamma^k_T)
\]
\end{mylemma}

\begin{proof}
Given $\ul \eta\in\R^{3N}$ we write $\ul \eta^{\leq k}=(\vec \eta_1,\ldots, \vec \eta_k)$ and $\ul \eta^{> k}=(\vec 
\eta_{k+1},\ldots, \vec \eta_N)$. We have
\[
\begin{aligned}
d_2\left(\mathcal R\bigl[f\Gamma^{N-k}_T\bigr],\bar f\Gamma^N_T\right)&=
\sup_{\ul\eta\not=0}\frac1{\|\ul \eta\|^2}\int_{\mathcal O}\left(f( (O\ul \eta)^{\leq k})-\bar f\widehat\Gamma^{k}_T( (O\ul 
\eta)^{\leq 
k})\right)\widehat\Gamma_T^{N-k}((O\ul \eta)^{>k})d\sigma(O)\\
\leq&d_2(f,\bar f\Gamma^k_T)\sup_{\ul\eta\not=0}\frac1{\|\ul \eta\|^2}\int_{\mathcal O}\|(O\ul \eta)^{\leq 
k}\|^2\widehat\Gamma_T^{N-k}((O\ul \eta)^{>k})d\sigma(O)\\
\leq&d_2(f,\bar f\Gamma^k_T)\sup_{\ul\eta\not=0}\frac1{\|\ul \eta\|^2}\int_{\mathcal O}\|(O\ul \eta)^{\leq 
k}\|^2d\sigma(O)
\end{aligned}
\]
Observe now that for every $i,j$ we have 
\[
\int_{\mathcal O}(O\ul \eta)_i^2d\sigma(O)=\int_{\mathcal O}(O\ul \eta)_j^2d\sigma(O)
\]
so that
\[
\int_{\mathcal O}\|(O\ul \eta)^{\leq k}\|^2d\sigma(O)=\frac kN \|\ul\eta\|^2\, .
\]
\end{proof}

%%%%%%%%%%%%%%%%%%%%%
%%%%%%%%%%%%%%%%%%%%%
%%%%%%%%%%%%%%%%%%%%%
%%%%%%%%%%%%%%%%%%%%%
%%%%%%%%%%%%%%%%%%%%%
%%%%%%%%%%%%%%%%%%%%%	

\end{document}